\documentclass[10pt,a4paper]{article}
\usepackage{amssymb}
\usepackage{amsmath}
\usepackage{amsthm}
\usepackage{latexsym}
\usepackage[dvips]{epsfig}
\usepackage{enumerate}
\usepackage{mathrsfs}
\usepackage{eufrak}
\usepackage{bm}
\usepackage{tikz}
\usepackage{authblk}
\usepackage{cancel}
\usepackage{cleveref}
\usepackage[utf8]{inputenc}

\DeclareRobustCommand{\rchi}{{\mathpalette\irchi\relax}}
\newcommand{\irchi}[2]{\raisebox{\depth}{$#1\chi$}}

\theoremstyle{plain}
\newtheorem{proposition}{Proposition}
\newtheorem{lemma}{Lemma}
\newtheorem{theorem}{Theorem}
\newtheorem{assumption}{Assumption}

\newtheorem{definition}{Definition}
\newtheorem{remark}{Remark}

\def\bmd{{\bm d}}

\def\bmg{{\bm g}}
\def\bmh{{\bm h}}
\def\bmi{{\bm i}}
\def\bmj{{\bm j}}

\def\bml{{\bm l}}

\def\bmm{{\bm m}}

\def\bmK{{\bm K}}
\def\bmL{{\bm L}}

\def\bmomega{{\bm \omega}}

\def\notSigma{{\not\! \Sigma}}
\def\notK{{\not\!\! K}}
\def\notL{{\not\!\!L}}
\def\notD{{\not\!\!D}}
\def\notd{{\not\!d}}
\def\notl{{\not\!l}}
\def\notalpha{{\not\!\alpha}}
\def\notbeta{{\not\!\beta}}
\def\notn{{\not\!n}}

\newcounter{mnotecount}

\newcommand{\mnotex}[1]
{\protect{\stepcounter{mnotecount}}$^{\mbox{\footnotesize $\bullet$\themnotecount}}$ 
\marginpar{
\raggedright\tiny\em
$\!\!\!\!\!\!\,\bullet$\themnotecount: #1} }

\begin{document}

\title{\textbf{Construction of anti-de Sitter-like
    spacetimes using the metric conformal Einstein field equations:
    the vacuum case}}

\author[,1]{Diego A. Carranza  \footnote{E-mail address:{\tt d.a.carranzaortiz@qmul.ac.uk}}}
\author[,1]{Juan A. Valiente Kroon \footnote{E-mail address:{\tt j.a.valiente-kroon@qmul.ac.uk}}}
\affil[1]{School of Mathematical Sciences, Queen Mary University of London,
Mile End Road, London E1 4NS, United Kingdom.}

\maketitle

\begin{abstract}
We make use of the metric version of the conformal Einstein field
equations to construct anti-de Sitter-like spacetimes by means of a
suitably posed initial-boundary value problem. The evolution system
associated to this initial-boundary value problem consists of a set of
conformal wave equations for a number of conformal fields and the
conformal metric. This formulation makes use of generalised wave
coordinates and allows the free specification of the Ricci scalar of
the conformal metric via a conformal gauge source function. We
consider Dirichlet boundary conditions for the evolution equations at
the conformal boundary and show that these boundary conditions can, in
turn, be constructed from the 3-dimensional Lorentzian metric of the
conformal boundary and a linear combination of the incoming and
outgoing radiation as measured by certain components of the Weyl
tensor.  To show that a solution to the conformal evolution equations
implies a solution to the Einstein field equations we also provide a
discussion of the propagation of the constraints for this
initial-boundary value problem. The existence of local solutions to
the initial-boundary value problem in a neighbourhood of the corner
where the initial hypersurface and the conformal boundary intersect is
subject to compatibility conditions between the initial and boundary
data. The construction described is amenable to numerical
implementation and should allow the systematic exploration of boundary
conditions.
\end{abstract}

\section{Introduction}
Anti-de Sitter-like spacetimes, i.e. spacetimes satisfying the
Einstein field equations with a negative Cosmological constant and
admitting a timelike conformal boundary, constitute a basic example of
solutions to the Einstein field equations which are not globally
hyperbolic ---see e.g. \cite{CFEBook}. As such, they cannot be
constructed solely from data on a
spacelike hypersurface and require the prescription of some suitable
boundary data. In view of the latter, the methods of conformal
geometry provide a natural setting for the discussion of an
initial-boundary value problem from which anti-de Sitter-like
spacetimes can be constructed in a systematic manner.  From the
conformal point of view, the timelike conformal boundary of the anti-de
Sitter-like spacetime has a finite location (described by the vanishing of
the conformal factor) so that analysis of the boundary
conditions and its relation to initial data can be carried out with
local computations.

\medskip A first analysis of the initial--boundary value problem for
4-dimensional vacuum anti-de Sitter-like spacetimes by means of conformal
methods has been carried out by Friedrich in \cite{Fri95} ---see also
\cite{Fri14} for further discussion of the admissible adS-like boundary
conditions. This seminal work makes use of a conformal representation of the
Einstein field equations known as the extended conformal Einstein field
equations and a gauge based on the properties of curves with good conformal
properties (\emph{conformal geodesics}) to set up an initial-boundary value
problem for a first order symmetric hyperbolic system of evolution equations. For this
type of evolution equations one can use the theory of maximally dissipative
boundary conditions as described in \cite{Rau85,Gue90} to assert the
well-posedness of the problem and to ensure the local existence of solutions in
a neighbourhood of the intersection of the initial hypersurface with the
conformal boundary (the \emph{corner}).  The solutions to these evolution
equation can be shown, via a further argument, to constitute a solution to the
vacuum Einstein field equations with negative Cosmological constant. 

Friedrich's analysis identifies a large class of maximally dissipative boundary
conditions involving the outgoing and incoming components of the Weyl tensor
---as such, they can be thought of as prescribing the relation between these
components.  These conditions are given in a very specific gauge and thus it is
difficult to assert their physical/geometric meaning. However, it is possible
to identify a subclass of boundary conditions which can be recast in a
covariant form. More precisely, they can be shown to be equivalent to
prescribing the conformal class of the metric on the conformal boundary ---see
\cite{Fri95}, also \cite{CFEBook}. The question of recasting the whole class of
maximally dissipative boundary conditions obtained by Friedrich in a geometric
(i.e. covariant) form remains an interesting open problem. 

An alternative construction of anti-de Sitter--like spacetimes, which does
not use the conformal Einstein field equations and which also hold for
spacetimes of dimension greater than four, can be found
in \cite{EncKam14}. A discussion of global properties of adS-like
spacetimes and the issue of their stability can be found in
\cite{And05b}. 

\medskip Numerical simulations involving anti-de Sitter-like spacetimes is a
very active area of current research ---see e.g.
\cite{BizRos11,BizMalRos15,DiaHorSan12,DiaHorMarSan12} which kick-started some
of the current flurry of interest. In this respect, Friedrich's results offer a
natural and systematic approach to the numerical construction of 4-dimensional
vacuum anti-de Sitter-like spacetimes. However, the numerical implementation of
these results is not straightforward, among other things, because the equations
involved are cast in a form which is not standard for the available
numerical codes and moreover, there is very little intuition about the
behaviour of the gauges used to formulate the equations. A further difficulty
of Friedrich's approach is that it cannot readily be extended to include matter
fields ---see however \cite{LueVal14a}.  

\medskip

In view of the issues raised in the previous paragraph, it is desirable to have
a conformal formulation of the initial-boundary value problem for anti de
Sitter-like spacetimes which is closer to the language used in numerical
simulations and which exploits familiar gauge conditions. In this article we
undertake this task. More precisely, we show that using the better know
\emph{metric conformal Einstein field equations} it is possible to construct
anti-de Sitter-like spacetimes by formulating an initial-boundary value problem
for a system of quasilinear wave equations for the conformal fields governing
the geometry of the conformal representation of the anti de Sitter-like
spacetimes. The partial differential equations (PDE) theory for this type of
systems is available in the literature \cite{CheWah83,DafHru85}. Dirichlet
boundary data for this system of wave equations can be constructed from the
prescription of the 3-dimensional (Lorentzian) metric of the conformal boundary
and a relation between the incoming and outgoing components of the Weyl tensor
akin to Friedrich's maximally dissipative conditions. Our setting makes use of
generalised harmonic coordinates. It also contains a further conformal gauge
freedom which can be fixed by specifying the value of the Ricci scalar of the
conformal representation. The evolution system to be solved can be thought of
as the Einstein field equations coupled to a complicated matter model
consisting of several tensorial fields ---each of which satisfies its own wave
equations. This parallel should ease the numerical implementation of the
setting. 

In addition to the formulation of an
initial-boundary value problem, we also study the relation between
the solutions to the conformal Einstein field equations and actual
solutions to the Einstein field equations. This analysis requires a
discussion similar to that of the \emph{propagation of the
  constraints} in which it is shown that a solution to the evolution
equations is, in fact, a solution to the original conformal field
equations. For this, one has to construct a subsidiary evolution system
for the conformal field equations and show that the boundary data
prescription for the evolution equations implies trivial boundary data
for the subsidiary equations. Fortunately, most of the lengthy
calculations required for this discussion are already available in the
literature ---see e.g. \cite{Pae15}.

\medskip
Our main result, stating the local existence in time of anti-de Sitter-like
spacetimes in a neighbourhood of the corner where the initial
hypersurface and the conformal boundary intersect, is provided in
Theorem \ref{Theorem:Main}. 

A feature of our analysis is that it can be extended to include
tracefree matter fields. This extension is discussed in a companion
article \cite{CarVal18c}.


\subsection*{Outline of the article}

This article is structured as follows: in Section \ref{Section:CFE} we
provide an overview of the relevant properties of our main technical
tool ---the so-called metric conformal Einstein field equations. In
particular we discuss the structural properties of the wave
equations describing the evolution of the conformal fields. We also
discuss the gauge fixing for the evolution equations and the key
properties of the subsidiary evolution system responsible of the
propagation of the constraints. In Section
\ref{Section:ConformalEinsteinConstraints} we discuss relevant
properties of the constraint equations implied on spacelike or timelike
hypersurfaces by the conformal Einstein field equations ---the
so-called conformal Einstein constraint equations. These constraint equations are
both 
relevant for the construction of suitable initial and boundary
initial data. In Section \ref{Section:GeneralSetUp} we describe the
general set-up of our construction of anti-de Sitter-like
spacetimes. In particular, we analyse the construction of suitable
boundary data directly from the knowledge of the metric at the
conformal boundary. We also study the properties of the compatibility
(corner) conditions between initial and boundary initial data required
to ensure the existence of solutions to the initial-boundary value
problem for the wave equations describing the evolution of the
conformal fields. In Section \ref{Section:PropagationOfTheConstraints}
we analyse the issue of the propagation of the constraints ---key to
establish the relation between solutions to the conformal wave
equations and actual solutions to the Einstein field equations. The
propagation of the constraints is established through the analysis of
a boundary-initial value problem for the subsidiary evolution
system. Finally, in Section \ref{Section:LocalExistenceResult} we
summarise our analysis by stating our main result, Theorem
\ref{Theorem:Main}. Section \ref{Section:Conclusions} provides some
concluding remarks to our analysis. In Appendix
\ref{Appendix:ZeroQuantities} we provide a discussion of the
integrability conditions associated to the metric conformal field
equations. These integrability conditions are fundamental to establish
a number of general properties of the conformal Einstein field equations. 

\subsection*{Conventions}
Throughout, the term \emph{spacetime} will be used to denote a 4-dimensional
Lorentzian manifold which not necessarily satisfies the Einstein field
equations. Moreover, $(\tilde{\mathcal{M}},\tilde{\bmg})$ will denote a vacuum
spacetime satisfying the Einstein equations with anti-de Sitter-like
cosmological constant $\lambda$.  The signature of the metric in this article
will be $(-,+,+,+)$. It follows that $\lambda<0$. The lowercase Latin letters
$a,\, b,\, c, \ldots$ are used as abstract spacetime tensor indices while the
indices $i,\,j,\,k,\ldots$ are abstract indices on the tensor bundle of
hypersurfaces of $\tilde{\mathcal{M}}$. The Greek letters $\mu, \, \nu, \,
\lambda,\ldots$ will be used as spacetime coordinate indices while
$\alpha,\,\beta,\,\gamma,\ldots$ will serve as indices on a hypersurface. Our
conventions for the curvature are
\[
\nabla_c \nabla_d u^a -\nabla_d \nabla_c u^a = R^a{}_{bcd} u^b.
\]

\section{The metric conformal Einstein field equations}
\label{Section:CFE}

The basic tool to be used in this article are the metric conformal
field equations. This section reviews the properties of this conformal
representation of the Einstein field equations that will be used
throughout.

\medskip
In what follows, let $(\tilde{\mathcal{M}},\tilde{g}_{ab})$ denote a spacetime
satisfying the vacuum Einstein field equations
\begin{equation}
\tilde{R}_{ab} = \lambda \tilde{g}_{ab},
\label{EFE}
\end{equation}
where $\tilde{R}_{ab}$ denotes the Ricci tensor of the metric
$\tilde{g}_{ab}$. Further, let $(\mathcal{M},g_{ab})$ denote a
spacetime conformally related to
$(\tilde{\mathcal{M}},\tilde{g}_{ab})$ so that, in a slight abuse of
notation, we have
\[
g_{ab} = \Xi^2 \tilde{g}_{ab},
\]
where $\Xi$ is some suitable conformal factor $\Xi$. The set of points
of $\mathcal{M}$ for which $\Xi$ vanishes will be called the
\emph{conformal boundary}. We use the notation $\mathscr{I}$ to denote
the parts of the conformal boundary which are a hypersurface of
$\mathcal{M}$. 

\subsection{Basic properties}
In what follows, let $\nabla_a$ denote the Levi-Civita connection of
the metric $g_{ab}$, and let $R^a{}_{bcd}$, $R_{ab}$, $R$,
$C^a{}_{bcd}$ denote, respectively, the associated Riemann tensor,
Ricci tensor, Ricci scalar and (conformally invariant) Weyl tensor. In the discussion of the
conformal Einstein field equations it is useful to introduce the
\emph{Schouten tensor}, defined as
\[
L_{ab} \equiv \frac{1}{2}\bigg(R_{ab} - \frac{1}{6} R g_{ab}\bigg).
\]
Moreover, let
\[
s \equiv\frac{1}{4} \nabla^c\nabla_c \Xi + \frac{1}{24}R\Xi, \qquad d^a{}_{bcd} \equiv \Xi^{-1}C^a{}_{bcd}
\]
denote the so-called \emph{Friedrich scalar} and the \emph{rescaled
  Weyl tensor}, respectively. 

\medskip
In terms of the objects defined in the previous paragraph, the
\emph{vacuum metric conformal Einstein field equations} are given by:
\begin{subequations}
\begin{eqnarray}
&& \nabla_a \nabla_b \Xi =-\Xi L_{ab} + sg_{ab}, \label{CFE1}\\
&& \nabla_a s = -L_{ac} \nabla^c \Xi, \label{CFE2}\\
&& \nabla_c L_{db} - \nabla_d L_{cd} = \nabla_a \Xi d^a{}_{bcd}, \label{CFE3}\\
&& \nabla_a d^a{}_{bcd}=0,  \label{CFE4}\\
&& 6 \Xi s - 3 \nabla_c \Xi \nabla^c \Xi =\lambda. \label{CFE5}
\end{eqnarray}
\end{subequations}

\begin{remark}
{\em Equations \eqref{CFE1}-\eqref{CFE4} will be read as differential
  conditions on the fields $\Xi$, $s$, $L_{ab}$, $d^a{}_{bcd}$ while
  equation \eqref{CFE5} will be regarded as a constraint which is
  satisfied if it holds at a single point by virtue of the other
  equations ---see Lemma 8.1 in \cite{CFEBook}.}
\end{remark}

\medskip
By a solution to the metric conformal Einstein field equations it is
understood a collection of fields 
\[
(g_{ab}, \Xi, s, L_{ab}, d^a{}_{bcd})
\]
satisfying equations \eqref{CFE1}-\eqref{CFE5}.  
 The relation between the metric conformal Einstein field equations and
the Einstein field equations is given by the following:

\begin{proposition}
\label{FriedrichThm}
Let  $(g_{ab}, \Xi, s, L_{ab}, d^a{}_{bcd})$ denote a solution to the
metric conformal Einstein field equations \eqref{CFE1}-\eqref{CFE4} such that $\Xi\neq 0$ on an
open set $\mathcal{U}\subset \mathcal{M}$. If, in addition, equation
\eqref{CFE5} is satisfied at a point $p\in \mathcal{U}$, then the
metric 
\[
\tilde{g}_{ab} = \Xi^{-2} g_{ab}
\]
 is a solution to the
Einstein field equations \eqref{EFE} on $\mathcal{U}$. 
\end{proposition}

A proof of the above proposition is given in \cite{CFEBook} ---see
Proposition 8.1 in that reference. 

\medskip
We also recall that the causal character of $\mathscr{I}$ is
determined by the sign of the Cosmological constant. More precisely,
one has that:

\begin{proposition}
Suppose that the Friedrich scalar $s$ is regular on $\mathscr{I}$.
Then $\mathscr{I}$ is a null, spacelike or timelike hypersurface of
$\mathcal{M}$, respectively, depending on whether $\lambda=0$, $\lambda>0$ or
$\lambda<0$. 
\end{proposition}

This result follows directly from evaluation of equation \eqref{CFE5}
on $\mathscr{I}$ and recalling that $\nabla_a\Xi$ is normal to this hypersurface.

\subsection{Wave equations for the conformal fields}
In \cite{Pae15} it has been shown how the conformal Einstein field
equations \eqref{CFE1}-\eqref{CFE4} imply a system of \emph{geometric}
wave equations for the components of the fields
$(\Xi,\,s,\,L_{ab},\,d^a{}_{bcd})$. In the subsequent discussion it
will be convenient to split the Schouten tensor into a tracefree part
and a pure-trace part (the Ricci scalar). Accordingly, one defines the
\emph{tracefree Ricci tensor} as
\[
\Phi_{ab} \equiv \frac{1}{2}\big( R_{ab} -\frac{1}{4} R g_{ab} \big),
\]
so that the Schouten tensor can be expressed as
\begin{equation}
L_{ab} = \Phi_{ab} + \frac{1}{24}R g_{ab}.
\label{SchoutenToTracefreeRicci}
\end{equation}

\medskip
In terms of the above field the main result in \cite{Pae15} can be
expressed as:

\begin{proposition}
Any solution $(\Xi,\,s,\,L_{ab},\,d^a{}_{bcd})$ to the conformal Einstein field
equations \eqref{CFE1}-\eqref{CFE4} satisfies the equations
\begin{subequations}
\begin{eqnarray}
&&\square \Xi = 4s -\frac{1}{6}\Xi R, \label{WaveCFE1}
  \\
&&\square s=  \Xi \Phi_{ab} \Phi^{ab}- \frac{1}{6} s R + \frac{1}{144} \Xi R^2 -  \frac{1}{6} \nabla_{a}R \nabla^{a}\Xi, \label{WaveCFE2}\\
&& \square \Phi_{ab}= 4 \Phi_{a}{}^{c} \Phi_{bc}-  g_{ab} \Phi_{cd} \Phi^{cd}-2 \Xi d_{acbd} \Phi^{cd} + \frac{1}{3} R \Phi_{ab} -  \frac{1}{24} g_{ab} \nabla_{c}\nabla^{c}R + \frac{1}{6} \nabla_{a}\nabla_{b}R, \label{WaveCFE3}\\
&& \square d_{abcd} = 2 \Xi d_{a}{}^{e}{}_{d}{}^{f} d_{becf} 
- 2 \Xi d_{a}{}^{e}{}_{c}{}^{f} d_{bedf} - 2 \Xi d_{ab}{}^{ef} 
d_{cedf} + \frac{1}{2} d_{abcd} R.\label{WaveCFE4}
\end{eqnarray}
\end{subequations}
\end{proposition} 

\begin{remark}
{\em The above wave equations are geometric, in the sense that they hold
  independently of the choice of coordinate system. However, as they
  stand they are not yet satisfactory  second order evolution equations to which one can apply
  the theory of partial differential equations. For this one has to
  provide a prescription of the Ricci scalar and introduce suitable
  coordinates. These issues are discussed in the following subsections.}
\end{remark}

\begin{remark}
{\em The wave equations \eqref{WaveCFE1}-\eqref{WaveCFE4} need to be
  supplemented with an equation for the components of the metric
  tensor $g_{ab}$. This equation is given by the definition of the
  tracefree Ricci tensor, equation \eqref{SchoutenToTracefreeRicci}, rewritten in the form
\begin{equation}
R_{ab} = 2 \Phi_{ab} + \frac{1}{4}Rg_{ab},
\label{SupplementaryEquation}
\end{equation}
where $R_{ab}$, and $\Phi_{ab}$ are regarded as independent
objects ---the former given through the classical expression in terms of
second order partial derivatives of the components of the metric
tensor while the latter as the field satisfying equations
\eqref{WaveCFE1}-\eqref{WaveCFE4}. 
}
\end{remark}

\subsection{Gauge considerations}
The conformal Einstein field equations possess both a coordinate and a
conformal freedom which can be exploited to cast the geometric wave
equations \eqref{WaveCFE1}-\eqref{WaveCFE4} as satisfactory hyperbolic
evolution equations.

\subsubsection{Conformal gauge source functions}
In the following, the Ricci scalar $R$ of the metric $g_{ab}$ will be
regarded as a \emph{conformal gauge source} specifying the
representative in the conformal class $[\tilde{\bmg}]$ one is working
with. Recall that given two conformally related metrics $g_{ab}$ and
$g'_{ab}$ such that $g'_{ab}=\vartheta^2 g_{ab}$, the respective Ricci
scalars are related to each other via
\[
R \vartheta - R'\vartheta^3 =6 \nabla^c \nabla_c \vartheta.
\]
If the values of $R$ and $R'$ are prescribed, the above transformation
law can be recast as a wave equation for the conformal factor relating
the two metrics. Namely, one has that
\[
\square \vartheta -\frac{1}{6}R\vartheta = -\frac{1}{6}R'\vartheta^3.
\]
Given suitable initial data for this wave equation, it can always be
solved locally. Accordingly, it is always possible to find (locally) a
conformal rescaling such that the metric $g'_{ab}$ has a prescribed
Ricci scalar $R'$. 

\begin{remark}
\label{Remark:ConformalGaugeSourceFunction}
{\em Following the previous discussion, in what follows the Ricci
  scalar of the metric $g_{ab}$ is regarded as a prescribed function $\mathcal{R}(x)$
  of the coordinates and one writes 
\[
R=\mathcal{R}(x).
\] }
\end{remark}

\subsubsection{Generalised harmonic coordinates and the reduced Ricci operator}
Given general coordinates $x=(x^\mu)$, the  components of the Ricci
tensor $R_{ab}$ can be explicitly written in terms of the components
of the metric tensor $g_{ab}$ and its first and second partial
derivatives as
\[
R_{\mu\nu} =
  -\frac{1}{2}g^{\lambda\rho} \partial_\lambda \partial_\rho
  g_{\mu\nu} + g_{\sigma(\mu}\nabla_{\nu)} \Gamma^\sigma + g_{\lambda\rho}
  g^{\sigma\tau} \Gamma^\lambda{}_{\sigma\mu} \Gamma^\rho{}_{\tau\nu}
  + 2 \Gamma^\sigma{}_{\lambda\rho} g^{\lambda\tau} g_{\sigma(\mu} \Gamma^\rho{}_{\nu)\tau},
\]
with
\[
\Gamma^\nu{}_{\mu\lambda} = \frac{1}{2}g^{\nu\rho} ( \partial_\mu
g_{\rho\lambda} + \partial_\lambda g_{\mu\rho} - \partial_\rho g_{\mu\lambda}),
\]
and where one has defined the \emph{contracted Christoffel symbols}
\[
\Gamma^\nu \equiv g^{\mu\lambda} \Gamma^\nu{}_{\mu\lambda}.
\]
A direct computation then gives
\[
\square x^\mu = - \Gamma^\mu.
\]

In what follows, we introduce  \emph{coordinate gauge source
  functions} $\mathcal{F}^\mu(x)$ to prescribe the value of the
contracted Christoffel symbols via the condition
\[
\Gamma^\mu = \mathcal{F}^\mu(x),
\]
so that the coordinates $x=(x^\mu)$ satisfy the \emph{generalised
  wave coordinate condition}
\begin{equation}
\square x^\mu = -\mathcal{F}^\mu(x).
\label{GeneralisedWaveCoordinates}
\end{equation}
Associated to the latter coordinate condition one then defines the \emph{reduced Ricci operator} $\mathscr{R}_{\mu\nu}[\bmg]$ as
\begin{equation}
\mathscr{R}_{\mu\nu}[\bmg] \equiv R_{\mu\nu} - g_{\sigma(\mu}\nabla_{\nu)}\Gamma^\sigma +
g_{\sigma(\mu}\nabla_{\nu)} \mathcal{F}^\sigma(x).
\label{DefinitionReducedRicci}
\end{equation}
More explicitly, one has that
\[
\mathscr{R}_{\mu\nu}[\bmg] =
  -\frac{1}{2}g^{\lambda\rho} \partial_\lambda \partial_\rho
  g_{\mu\nu} - g_{\sigma(\mu}\nabla_{\nu)} \mathcal{F}^\sigma(x) + g_{\lambda\rho}
  g^{\sigma\tau} \Gamma^\lambda{}_{\sigma\mu} \Gamma^\rho{}_{\tau\nu}
  + 2 \Gamma^\sigma{}_{\lambda\rho} g^{\lambda\tau} g_{\sigma(\mu} \Gamma^\rho{}_{\nu)\tau}.
\]
Thus, by choosing coordinates satisfying the \emph{generalised
  wave coordinates condition} \eqref{GeneralisedWaveCoordinates}, the
\emph{unphysical Einstein equation} \eqref{SupplementaryEquation} takes the form 
\begin{equation}
\mathscr{R}_{\mu\nu}[\bmg] = 2 \Phi_{\mu\nu} +\frac{1}{4}\mathcal{R}(x) g_{\mu\nu}.
\label{UnphysicalEinsteinEquation}
\end{equation}
Assuming that the components $\Phi_{\mu\nu}$ are known, the latter is a quasilinear wave
equation for the components of the metric tensor. 

\subsubsection{The reduced wave operator}
While equations \eqref{WaveCFE1} and \eqref{WaveCFE2} provide
satisfactory wave equations for the scalar fields $\Xi$ and $s$
independently of the choice of coordinates, this is not the case for
equations \eqref{WaveCFE3} and \eqref{WaveCFE4}. The reason for this
is that in these equations the wave operator $\square$ is acting on
tensors, and thus, the terms $\square \Phi_{ab}$ and $\square
d_{abcd}$, when expressed in a given coordinate system $x=(x^\mu)$,
involve derivatives of Christoffel symbols ---and consequently, second order
derivatives of the metric tensor. This is a problem in situations,
like the one considered here, where the metric is an unknown in the
problem as the presence of these derivatives in the operator destroys the
hyperbolicity of the system.

\medskip
In what follows, it will be shown how the generalised wave coordinate
condition \eqref{GeneralisedWaveCoordinates} can be used to reduce the
geometric wave operator $\square$ to a second order hyperbolic
operator. To motivate the procedure consider a covector $\omega_a$
with components $\omega_\mu$ with respect to a coordinate system
$x=(x^\mu)$ satisfying condition \eqref{GeneralisedWaveCoordinates}
for some choice of coordinate gauge source functions
$\mathcal{F}^\mu(x)$. A direct computation using the expression of the covariant derivative
in terms of Christoffel symbols yields
\begin{eqnarray*}
&& \square \omega_{\lambda} \equiv g^{\mu\nu}\nabla_\mu \nabla_\nu
   \omega_{\lambda} \\
&& \phantom{\square \Phi_{\lambda}} =
   g^{\mu\nu} \partial_\mu\partial_\nu \omega_{\lambda}
   -g^{\mu\nu} \partial_\mu \Gamma^\sigma{}_{\nu\lambda}
   \omega_{\sigma} + f_{\lambda}(g,\partial g,\omega,\partial\omega),
\end{eqnarray*}
where $f_{\lambda}(g,\partial g,\omega,\partial\omega)$ denotes an
expression depending on the components $g_{\mu\nu}$, $\omega_{\mu}$ and
their first order partial derivatives. Now, recall the classical
expression for the components of the Riemann tensor in terms of the
Christoffel symbols and their derivatives, 
\[
R^\sigma{}_{\mu\lambda\nu} = \partial_\lambda \Gamma^\sigma{}_{\nu\mu}
-\partial_\nu \Gamma^\sigma{}_{\lambda\mu} + \Gamma^\sigma{}_{\lambda\tau}\Gamma^\tau{}_{\nu\mu}-\Gamma^\sigma{}_{\nu\tau}\Gamma^\tau{}_{\lambda\mu}
\]
so that
\begin{eqnarray*}
&& R^\sigma{}_\lambda = g^{\mu\nu} R^\sigma{}_{\mu\lambda\nu} \\
&& \phantom{R^\sigma{}_\lambda} = g^{\mu\nu} \partial_\lambda
   \Gamma^\sigma{}_{\nu\mu} - g^{\mu\nu} \partial_\nu
   \Gamma^\sigma{}_{\lambda\mu} + g^{\mu\nu} \Gamma^\sigma{}_{\lambda\tau}\Gamma^\tau{}_{\nu\mu}-g^{\mu\nu}\Gamma^\sigma{}_{\nu\tau}\Gamma^\tau{}_{\lambda\mu}.
\end{eqnarray*}
Making use of this coordinate expression on obtains
\begin{eqnarray*}
&& \square \omega_{\lambda} = g^{\mu\nu}\partial_\mu\partial_\nu
  \omega_{\lambda} + \big( R^\sigma{}_\lambda -g^{\mu\nu}\partial_\lambda
  \Gamma^\sigma{}_{\nu\mu}  \big)\omega_{\sigma}
  + f_{\lambda}
  (\bmg,\partial \bmg,\bmomega,\partial \bmomega) \\
&& \phantom{\square \omega_{\lambda}} = g^{\mu\nu}\partial_\mu\partial_\nu
  \omega_{\lambda} + \big( R^\sigma{}_\lambda-\partial_\lambda
  \Gamma^\sigma \big)\omega_{\sigma}+ f_{\lambda}
  (\bmg,\partial \bmg,\bmomega,\partial \bmomega)\\
&& \phantom{\square \omega_{\lambda}} = g^{\mu\nu}\partial_\mu\partial_\nu
  \omega_{\lambda} + \big( R_{\tau\lambda}-g_{\sigma\tau}\partial_\lambda
  \Gamma^\sigma \big)\omega{}^\tau{}+ f_{\lambda}
  (\bmg,\partial \bmg,\bmomega,\partial \bmomega),
\end{eqnarray*}
and finally
\begin{equation}
 \square \omega_{\lambda}= g^{\mu\nu}\partial_\mu\partial_\nu
  \omega_{\lambda} + \big( R_{\tau\lambda}-g_{\sigma\tau}\nabla_\lambda
  \Gamma^\sigma \big)\omega^\tau+ f_{\lambda}
  (\bmg,\partial \bmg,\bmomega,\partial \bmomega).
\label{Massaged:Boxomega}
\end{equation}
Making the formal replacements
\[
R_{\mu\nu}\mapsto 2\Phi_{\mu\nu}+ \frac{1}{4} \mathcal{R}(x)g_{\mu\nu}, \qquad \Gamma^\mu \mapsto \mathcal{F}^\mu(x), 
\]
in equation \eqref{Massaged:Boxomega}, one defines the \emph{reduced wave
  operator} $\blacksquare$, acting on the components $\omega_{\mu}$ as
\begin{eqnarray}
&&\blacksquare \omega_{\lambda} \equiv g^{\mu\nu}\partial_\mu\partial_\nu
  \omega_{\lambda}  + \bigg( 2 \Phi_{\tau\lambda} + \frac{1}{4}\mathcal{R}(x)
   g_{\tau\lambda} -
  g_{\sigma\tau}\nabla_\lambda \mathcal{F}^\sigma(x) \bigg)\omega^\tau + f_{\lambda}
  (\bmg,\partial \bmg,\bmomega,\partial \bmomega), \label{ReducedWaveOperator}
\end{eqnarray}
where $f_{\lambda}
  (\bmg,\partial \bmg,\bmomega,\partial \bmomega)$ denotes lower order terms
  whose explicit form will not be required. In fact, from the previous
  discussion it follows that one can write
\[
\blacksquare \omega_\lambda = \square \omega_\lambda + \bigg( (2 \Phi_{\tau\lambda} + \frac{1}{4}\mathcal{R}(x)
   g_{\tau\lambda} - R_{\tau\lambda}) -
  g_{\sigma\tau}\nabla_\lambda (\mathcal{F}^\sigma(x)-\Gamma^\sigma) \bigg)\omega^\tau.
\]

A similar construction for covariant tensors of arbitrary rank leads
to the following:

\begin{definition}
The reduced wave operator $\blacksquare$ acting on a covariant tensor field
$T_{\lambda\cdots\rho}$ is defined as
\begin{eqnarray*}
&& \blacksquare T_{\lambda \cdots \rho} \equiv \square
T_{\lambda\cdots\rho} +\bigg( (2 \Phi_{\tau\lambda} + \frac{1}{4}\mathcal{R}(x)
   g_{\tau\lambda} - R_{\tau\lambda}) -
  g_{\sigma\tau}\nabla_\lambda (\mathcal{F}^\sigma(x)-\Gamma^\sigma) \bigg)T^\tau{}_{\cdots\rho}
 +\cdots\\
&& \hspace{3cm} \cdots + \bigg( (2 \Phi_{\tau\rho} + \frac{1}{4}\mathcal{R}(x)
   g_{\tau\rho} - R_{\tau\rho}) -
  g_{\sigma\tau}\nabla_\rho (\mathcal{F}^\sigma(x)-\Gamma^\sigma)
  \bigg)T_{\lambda\cdots}{}^\tau
\end{eqnarray*}
where $\square \equiv g^{\mu\nu}\nabla_\mu\nabla_\nu$. The action of
$\blacksquare$ on a scalar $\phi$ is simply given by
\[
\blacksquare \phi \equiv g^{\mu\nu}\nabla_\mu\nabla_\nu \phi. 
\]
\end{definition}

\begin{remark}
{\em The operator $\blacksquare$ provides a proper second order hyperbolic
operator ---in contrast to $\square$.  Accordingly, when working in generalised
harmonic coordinates, all the second order derivatives of the metric tensor can
be removed from the principal part of the evolution equations
\eqref{WaveCFE3} and \eqref{WaveCFE4}. }
\end{remark}

\subsubsection{Summary: gauge reduced evolution equations}
The discussion of the previous sections leads us to consider the
following \emph{gauge reduced} system of evolution equations for the
components of the conformal fields $\Xi$, $s$, $\Phi_{ab}$,
$d_{abcd}$ and $g_{ab}$ with respect to coordinates $x=(x^\mu)$
satisfying the generalised wave coordinate condition \eqref{GeneralisedWaveCoordinates}:
\begin{subequations}
\begin{eqnarray}
&&\blacksquare \Xi = 4s -\frac{1}{6}\Xi \mathcal{R}(x), \label{ReducedWaveCFE1}
  \\
&&\blacksquare s=  \Xi \Phi_{\mu\nu} \Phi^{\mu\nu}- \frac{1}{6} s \mathcal{R}(x) + \frac{1}{144} \Xi \mathcal{R}(x)^2 -  \frac{1}{6} \nabla_{\mu}\mathcal{R}(x) \nabla^{\mu}\Xi, \label{ReducedWaveCFE2}\\
&& \blacksquare \Phi_{\mu\nu}= 4 \Phi_{\mu}{}^{\lambda} \Phi_{\nu\lambda}-  g_{\mu\nu} \Phi_{\lambda\rho}
\Phi^{\lambda\rho}-2 \Xi d_{\mu\lambda\nu\rho} \Phi^{\lambda\rho} \nonumber \\
&& \hspace{3cm}+ \frac{1}{3} \mathcal{R}(x) \Phi_{\mu\nu} -  \frac{1}{24} g_{\mu\nu} \nabla_{\lambda}\nabla^{\lambda}\mathcal{R}(x) + \frac{1}{6} \nabla_{\mu}\nabla_{\nu}\mathcal{R}(x), \label{ReducedWaveCFE3}\\
&& \blacksquare d_{\mu\nu\lambda\rho} = 2 \Xi d_{\mu}{}^{\sigma}{}_{\rho}{}^{\tau} d_{\nu\sigma\lambda\tau} 
- 2 \Xi d_{\mu}{}^{\sigma}{}_{\lambda}{}^{\tau} d_{\nu\sigma\rho\tau} - 2 \Xi d_{\mu\nu}{}^{\sigma\tau} 
d_{\lambda\sigma\rho\tau} + \frac{1}{2} d_{\mu\nu\lambda\rho} \mathcal{R}(x),\label{ReducedWaveCFE4}\\
&& \mathscr{R}_{\mu\nu}[\bmg] = 2 \Phi_{\mu\nu} +\frac{1}{4}\mathcal{R}(x)
g_{\mu\nu}. \label{ReducedWaveCFE5}
\end{eqnarray}
\end{subequations}

\begin{remark}
{\em The reduced system
  \eqref{ReducedWaveCFE1}-\eqref{ReducedWaveCFE5} constitutes a system
of quasilinear wave equations for the fields $\Xi$, $s$,
$\Phi_{\mu\nu}$, $d_{\mu\nu\lambda\rho}$ and $g_{\mu\nu}$. More
explicitly, one has that
\begin{eqnarray*}
&& g^{\sigma\tau}\partial_\sigma\partial_\tau \Xi = X
\big(\bmg,\partial\bmg, \Xi,s,\mathcal{R}(x)\big), \\
&& g^{\sigma\tau}\partial_\sigma\partial_\tau s =
S\big(\bmg,\partial\bmg,\Xi, \partial\Xi,s,{\bm
  \Phi},\mathcal{R}(x),\partial\mathcal{R}(x) \big), \\
&& g^{\sigma\tau}\partial_\sigma\partial_\tau \Phi_{\mu\nu}
=F_{\mu\nu}\big(\bmg,\partial\bmg,\Xi,{\bm \Phi},{\bm d},
\mathcal{R}(x),\partial^2\mathcal{R}(x)\big), \\
&&  g^{\sigma\tau}\partial_\sigma\partial_\tau
d_{\mu\nu\lambda\rho}=D_{\mu\nu\lambda\rho}\big(\bmg,\partial\bmg,\Xi,{\bm
  d},\mathcal{R}(x)\big), \\
&& g^{\sigma\tau}\partial_\sigma\partial_\tau g_{\mu\nu}
=G_{\mu\nu}\big(\bmg,\partial\bmg,{\bm \Phi},\mathcal{R}(x)\big),
\end{eqnarray*}
where $X$, $S$, $F_{\mu\nu}$, $D_{\mu\nu\lambda\rho}$ and $G_{\mu\nu}$
are polynomial expressions of their arguments. Strictly speaking, the system is a system of wave
  equations only if $g_{\mu\nu}$ is known to be Lorentzian. This will
  be case in a perturbative setting or close to an initial
  hypersurface where initial data can be prescribed to this
  effect. The local existence theory of initial-boundary value problems for systems of
quasilinear differential equations of the above type with Dirichlet
boundary data can be found in e.g. \cite{CheWah83,DafHru85}. }
\end{remark}

\subsection{The subsidiary evolution equations}
\label{zero_quantities}
In order to analyse the relation between solutions to the system of
geometric wave equations \eqref{WaveCFE1}-\eqref{WaveCFE4} and the
conformal Einstein field equations \eqref{CFE1}-\eqref{CFE5} one needs
to construct a subsidiary system of equations encoding the evolution
of these equations. Accordingly, one defines the \emph{zero-quantities}
\begin{subequations}
\begin{eqnarray}
&& \Upsilon_{ab} \equiv \nabla_a \nabla_b \Xi + \Xi L_{ab} - sg_{ab}, \label{SubsidiaryDefinition1}\\
&& \Theta_a \equiv \nabla_a s + L_{ac} \nabla^c \Xi, \label{SubsidiaryDefinition2}\\
&& \Delta_{cdb} \equiv \nabla_c L_{db} - \nabla_d L_{cd} + \nabla_a \Xi
d^a{}_{bcd}, \label{SubsidiaryDefinition3}\\
&& \Lambda_{bcd}\equiv \nabla_a d^a{}_{bcd}. \label{SubsidiaryDefinition4}
\end{eqnarray}
\end{subequations}
In terms of the latter, the conformal Einstein field equations
\eqref{CFE1}-\eqref{CFE4} can be expressed as
\begin{equation}
\Upsilon_{ab}=0, \qquad \Theta_a =0, \qquad \Delta_{cdb}=0, \qquad
\Lambda_{bcd}=0.
\label{CFEZeroQuantities}
\end{equation}

A lengthy computation, best done using computer algebra, leads to the
following: 

\begin{proposition}
\label{PropagationZQ}
Assume that the conformal fields $\Xi$, $s$, $L_{ab}$ and $d_{abcd}$
satisfy the geometric wave equations
\eqref{WaveCFE1}-\eqref{WaveCFE4}. Then the zero-quantities $\Theta_a$,
$\Upsilon_{ab}$, $\Delta_{abc}$ and $\Lambda_{abc}$ satisfy a 
system of geometric wave equations of the form
\begin{subequations}
\begin{eqnarray}
&& \square \Theta_a = H_a({\bm \Theta},{\bm \Upsilon},{\bm
  \Delta},{\bm \Lambda}), \label{SubsidiaryEquation1}\\
&& \square \Upsilon_{ab} = H_{ab}({\bm \Theta},{\bm \Upsilon}, {\bm\nabla\bm\Upsilon}, {\bm
  \Delta}), \label{SubsidiaryEquation2}\\
&& \square \Delta_{abc} = H_{abc}({\bm
  \Delta},{\bm \Lambda}), \label{SubsidiaryEquation3}\\
&& \square \Lambda_{abc} = L_{abc}({\bm \Theta},{\bm \Upsilon},{\bm
  \Delta},{\bm \Lambda}), \label{SubsidiaryEquation4}
\end{eqnarray}
\end{subequations}
where $H_a$, $H_{ab}$, $H_{abc}$ and $L_{abc}$ are homogeneous
expressions of their arguments.
\end{proposition}

The original proof of this result was given in \cite{Pae15}. An alternative
derivation, along with several properties of the zero--quantities, can be found
in Appendix A.

\begin{remark}
{\em In practice, the geometric wave equations
\eqref{SubsidiaryEquation1}-\eqref{SubsidiaryEquation4} are replaced
by standard wave equations for the components of the zero fields by
exchanging the wave operator $\square$ by the reduced wave operator
$\blacksquare$.}
\end{remark}

\section{The conformal Einstein constraint equations}
\label{Section:ConformalEinsteinConstraints}
In order to formulate an initial-boundary value problem for the wave
equations \eqref{WaveCFE1}-\eqref{WaveCFE4} we will need the
constraint equations implied by the conformal Einstein field equations
\eqref{CFE1}-\eqref{CFE5} on (spacelike and timelike) hypersurfaces of
the unphysical spacetime $(\mathcal{M},g_{ab})$. These equations were
first discussed in \cite{Fri84}. A detailed discussion of their
derivation and basic properties can be found in \cite{CFEBook},
Chapter 11. 

\subsection{The basic expression of the conformal Einstein constraint
  equations}
Let $\mathcal{S}$ denote a (spacelike or timelike) hypersurface of the
unphysical spacetime $(\mathcal{M},g_{ab})$ with unit normal
$n_a$. Furthermore, let 
\[
\epsilon\equiv n_a n^a,
\]
 so that $\epsilon=1$
if $\mathcal{S}$ is timelike and $\epsilon=-1$ if it is spacelike. The
projector to $\mathcal{S}$ is defined as 
\[
h_{ab} \equiv g_{ab}- \epsilon n_a n_b.
\]
The extrinsic curvature of $\mathcal{S}$ is defined as
\[
K_{ab} = h_a{}^c h_b{}^d \nabla_c n_d.
\]
The restriction of the conformal factor $\Xi$ to the hypersurface will be
denoted by $\Omega$.

\medskip
In the following let 
\[
\Sigma, \quad s, \quad h_{ij} \quad L_i, \quad L_{ij}, \quad d_{ij},
\quad d_{ijk}, \quad d_{ijkl} 
\]
denote, respectively, the pull-backs of 
\begin{eqnarray*}
&n^a \nabla_a \Xi, \quad s, \quad g_{ab}, \quad n^c h_a{}^d L_{cd},
  \quad h_a{}^ch_b{}^d L_{cd},& \\
& n^b n^d h_e{}^a h_f{}^cd_{abcd}, \quad n^b h_e{}^a h_f{}^c
h_g{}^d d_{abcd}, \quad h_e{}^ah_f{}^bh_g{}^ch_h{}^d d_{abcd} & 
\end{eqnarray*}
to $\mathcal{S}$. 

\begin{remark}
{\em In particular, $h_{ij}$ corresponds is the 3-metric
induced by $g_{ab}$ on $\mathcal{S}$. Similarly, we will denote by $K_{ij}$ the pull-back
of $K_{ab}$ and $K=h^{ij}K_{ij}$. The metric $h_{ij}$ will be
Lorentzian or Riemannian depending on whether $S$ is timelike or spacelike.}
\end{remark}

\begin{remark}
{\em The fields $d_{ij}$ and $d_{ijk}$ encode, respectively, the
\emph{electric} and \emph{magnetic parts} of the rescaled Weyl tensor
$d_{abcd}$ with respect to the normal $n_a$. It can be verified that
\begin{eqnarray*}
& d_i{}^i=0, \quad d_{ij}=d_{ji}, \quad d_{ijk}=-d_{ikj}, \quad d_{[ijk]}=0,
  & \\
& d_{ijkl} = 2 \epsilon(h_{i[l} d_{k]j} + h_{j[k}d_{l]i}).
\end{eqnarray*}
The magnetic part is more commonly encoded in a symmetric traceless
tensor of rank 2 defined as
\[
d^*_{ij} \equiv \frac{1}{2}\epsilon_j{}^{kl} d_{ikl}
\]
where $\epsilon_{ijk}$ is the volume form induced on $\mathcal{S}$ by
$h_{ij}$. }
\end{remark}

In terms of the above fields, a long computation shows that the conformal Einstein field equations
\eqref{CFE1}-\eqref{CFE5} imply on the hypersurface $\mathcal{S}$ the
\emph{conformal Einstein constraint equations}
\begin{subequations}
\begin{eqnarray}
&& D_i D_j \Omega = -\epsilon \Sigma K_{ij} -\Omega L_{ij} + s
h_{ij}, \label{ConformalConstraint1} \\
&& D_i \Sigma = K_i{}^k D_k \Omega -\Omega L_i, \label{ConformalConstraint2} \\
&& D_i s = -\epsilon L_i \Sigma - L_{ik} D^k \Omega, \label{ConformalConstraint3}\\
 && D_i L_{jk} -D_j L_{ik} = -\epsilon \Sigma d_{kij}
+ D^l \Omega d_{lkij} -\epsilon (K_{ik} L_j
- K_{jk} L_i), \label{ConformalConstraint4}\\
 && D_i L_j - D_j L_i = D^l \Omega d_{lij} +
K_i{}^k L_{jk} - K_j{}^k L_{ik}, \label{ConformalConstraint5}\\
&& D^k d_{kij} =\epsilon \big(K^k{}_i d_{jk}
   -K^k{}_j d_{ik}\big), \label{ConformalConstraint6}\\
&& D^i d_{ij}= K^{ik} d_{ijk}, \label{ConformalConstraint7}\\
&& \lambda = 6 \Omega s - 3\epsilon \Sigma^2 - 3 D_k \Omega D^k\Omega.
\label{ConformalConstraint8}
\end{eqnarray}
\end{subequations}

\noindent These equations are supplemented by the Codazzi-Mainardi and Gauss-Codazzi equations
which, respectively, take the following form:
\begin{subequations}
\begin{eqnarray}
&& D_j K_{ki} - D_k K_{ji} = \Omega d_{ijk} + h_{ij} L_k -
 h_{ik}L_j, \label{ConformalConstraint9} \\
&& l_{ij} = -\epsilon\Omega d_{ij} + L_{ij} + \epsilon \bigg( K \big(
 K_{ij} -\displaystyle\frac{1}{4} K h_{ij}\big) - K_{ki}
  K_j{}^k + \displaystyle\frac{1}{4}  K_{kl} K^{kl}h_{\bmi\bmj}\bigg),
  \label{ConformalConstraint10}
\end{eqnarray}
\end{subequations}
where the Schouten tensor of $h_{ij}$ is defined as
\[
l_{ij} \equiv r_{ij} -\frac{1}{4}r h_{ij}.
\]
Here, $r_{ij}$ and $r$
are, respectively, the Ricci tensor and scalar of the metric $h_{ij}$.

\subsection{The conformal constraints on the conformal boundary}
\label{Section:ConformalConstraintsConformalBoundary}
The conformal Einstein constraint equations simplify considerably when they are evaluated
on an hypersurface corresponding to the conformal boundary of a spacetime, in
which case $\Omega$ vanishes identically. If the conformal boundary is timelike
($\epsilon=1$) one has the following system:
\begin{subequations}
\begin{eqnarray}
&& s \ell_{ij} \simeq \notSigma \notK_{ij}, \label{CCCB1}\\
&& \notD_i \notSigma \simeq 0, \label{CCCB2}\\
&& \notD_i s \simeq - \notL_i \notSigma, \label{CCCB3}\\
&& \notD_i \notL_{jk} - \notD_j L_{ik} \simeq -\notSigma \notd_{kij} + ( \notK_{jk}\notL_i-\notK_{ik}\notL_j 
 ), \label{CCCB4}\\
&& \notD_i \notL_j - \notD_j \notL_i \simeq \notK_i{}^k \notL_{jk} - \notK_j{}^k \notL_{ik}, \label{CCCB5}\\
&& \notD^k \notd_{kij} \simeq \notK^k{}_j \notd_{ik} - \notK^k{}_i \notd_{jk}, \label{CCCB6}\\
&& \notD^i \notd_{ij}\simeq \notK^{ik} \notd_{ijk}, \label{CCCB7}\\
&& \lambda \simeq -3 \notSigma^2, \label{CCCB8}\\
&& \notD_j \notK_{ki} - \notD_k \notK_{ji}\simeq \ell_{ij} \notL_k - \ell_{ik} \notL_j, \label{CCCB9}\\
&& \notl_{ij} \simeq \notL_{ij} + \notK\big( \notK_{ij} - \frac{1}{4} \notK \ell_{ij} \big) -
   \notK_{ki}\notK_j{}^k + \frac{1}{4}\notK_{kl}\notK^{kl} \ell_{ij}, \label{CCCB10}
\end{eqnarray}
\end{subequations}
where $\simeq$ denotes that the equality holds on the conformal boundary and
$\ell_{ij}$ denotes the intrinsic (Lorentzian) 3-metric on
$\mathscr{I}$. Moreover, we use the notation $\not{\;}$
to indicate that the quantities are obtained from a 3+1 split with
respect to the (timelike) conformal boundary. In particular, $\notD_i$ denotes
the Levi-Civita connection of the Lorentzian metric $\ell_{ij}$.  \emph{This
notation will be used in the rest of the article.}

\medskip
In \cite{Fri86c} a procedure to solve the conformal constraints on the
conformal boundary has been given. The key observation is to identify
the scalar $s$ as gauge dependent quantity and the 3-metric
$\ell_{ij}$ on $\mathscr{I}$ as free data. Instead of directly working
with $s$ it is more convenient to consider a scalar $\varkappa$ such
that
\[
s \simeq \notSigma \varkappa.
\]

One has then that:
\begin{proposition}
\label{Proposition:ConformalConstraintsConformalBoundary}
Given a 3-dimensional Lorentzian metric $\ell_{ij}$, a
$\ell$-divergencefree and tracefree field $\notd_{ij}$ and a smooth
function $\varkappa$, then the fields
\begin{subequations}
\begin{eqnarray}
&& \notSigma \simeq \sqrt{\frac{|\lambda|}{3}},  \label{SolutionConstraints5}\\
&& s \simeq \notSigma \varkappa, \label{SolutionConstraints0}\\
&& \notK_{ij} \simeq \varkappa \ell_{ij}, \label{SolutionConstraints1}\\
&&  \notL_i \simeq -\notD_i \varkappa, \label{SolutionConstraints2}\\
&& \notL_{ij} \simeq \notl_{ij} - \frac{1}{2}\varkappa^2
   \ell_{ij}, \label{SolutionConstraints3}\\
&& \notd_{ijk} \simeq -\notSigma^{-1} y_{ijk}, \label{SolutionConstraints4}
\end{eqnarray}
\end{subequations}
where
\[
y_{ijk} \equiv \notD_j \notl_{ki}-\notD_k \notl_{ji}
\]
is the Cotton tensor of $\ell_{ij}$, constitute a solution to the
conformal constraint equations  \eqref{CCCB1}-\eqref{CCCB10} with
$\epsilon=1$ and $\Omega = 0$. 

\end{proposition}

A proof of this result can be found in \cite{CFEBook}, Section
11.4.4. We will also require the following partial converse the
previous result:


\begin{proposition}
Assume one has a timelike hypersurface $\mathcal{T}$ of a spacetime
$(\mathcal{M},g_{ab})$ such that conditions
\eqref{SolutionConstraints1}-\eqref{SolutionConstraints3} hold. If, in
addition, $\Omega=0$ on some fiduciary spacelike hypersurface
$\mathcal{C}_\star$ of $\mathcal{T}$ then one has that
\[
\Omega =0 \quad \mbox{on} \quad \mathcal{T}.
\]
\end{proposition}

\begin{proof}
Assume first that $\varkappa\neq 0$ on $\mathcal{C}_\star$. The
substitution of expressions
\eqref{SolutionConstraints1}-\eqref{SolutionConstraints3} into the
general conformal constraint equations \eqref{ConformalConstraint1}
and \eqref{ConformalConstraint3} yields the relations
\begin{subequations}
\begin{eqnarray}
&& \notD_i \notD_j \Omega = -\Omega\big( \notl_{ij} - \frac{1}{2}\varkappa^2
   \ell_{ij}\big), \label{ForcedConstraint1}\\
&& \varkappa \notD_i \Omega =\Omega \notD_i \varkappa. \label{ForcedConstraint2}
\end{eqnarray}
\end{subequations} 
Taking the trace of equation \eqref{ForcedConstraint1} one obtains the 
wave equation
\begin{equation}
\square_{\ell} \Omega = -\frac{1}{4} \big( \not\!r-6 \varkappa^2  \big)
\Omega
\label{WaveEquationOmegaScri}
\end{equation}
on $\mathcal{T}$, where $\square_{\ell}\equiv \ell^{ij} \notD_i
\notD_j$. We now consider the condition $\Omega=0$ on $\mathcal{C}_\star$
as initial data for equation \eqref{WaveEquationOmegaScri}. We
complement this initial condition with $\notD_i \Omega=0$ on
$\mathcal{C}_\star$ which follows from condition
\eqref{ForcedConstraint2}. It follows from the homogeneity of equation
\eqref{WaveEquationOmegaScri} and the uniqueness of
solutions to wave equations of this form that $\Omega=0$ on
$\mathcal{T}$. 

\smallskip
To deal with the case $\varkappa=0$ we observe that it is always
possible to carry out a rescaling $\Xi \mapsto \Xi' \equiv \vartheta
\Xi$ of the spacetime conformal factor $\Xi$ with $\vartheta\simeq 1$
and $\mathbf{d}\vartheta \neq 0$ such that if $s\not \simeq 0$ on
$\mathscr{I}$ then $s'\simeq 0$ ---see \cite{CFEBook} Section 11.4.4,
page 268. Thus, if $\varkappa\not \simeq 0$ initially, then using the
above rescaling and taking into account relation
\eqref{SolutionConstraints0} for $s'$, it follows that $\varkappa'
\simeq 0$. The rescaling $\Xi \mapsto \Xi' \equiv \vartheta \Xi$ does
not change the value of $\Xi$ on $\mathcal{T}$ ---accordingly one has
that $\Omega=0$ on $\mathcal{T}$ even if $\varkappa=0$.
\end{proof}

\subsection{Solutions to the conformal constraints on a spacelike hypersurface}
In addition to analysing the conformal constraint equations on a timelike
hypersurface corresponding to the conformal boundary of the spacetime, we will
also need to consider solutions to the constraints
\eqref{ConformalConstraint1}-\eqref{ConformalConstraint10} on spacelike
hypersurfaces. These solutions provide part of the initial data for the wave
equations \eqref{ReducedWaveCFE1}-\eqref{ReducedWaveCFE4}.

\medskip

The conformal constraint equations
\eqref{ConformalConstraint1}-\eqref{ConformalConstraint10} with $\epsilon=-1$
can be combined to obtain the \emph{conformal Hamiltonian} and \emph{momentum
constraints}
\begin{subequations}
\begin{eqnarray}
&& \lambda = -\frac{1}{2}\Omega^2 K_{ij} K^{ij} +  \frac{1}{2} \Omega^2 K^2 + \frac{1}{2} \Omega^2 r - 2  \Omega K \Sigma + 
3  \Sigma^2 - 3 D_{i}\Omega D^{i}\Omega + 2 \Omega
   D_{i}D^{i}\Omega,\\
&& \Omega D^jK_{ij} -  \Omega D_{i}K = 2 K_{ij} D^j\Omega- 2 D_{i}\Sigma.
\end{eqnarray}
\end{subequations}
For a solution of the above equations it will be understood a
collection of fields $(\Omega,h_{ij},K_{ij},\Sigma)$ satisfying them. The
collection $(\Omega,h_{ij},K_{ij},\Sigma)$ constitutes the \emph{basic data}
from which the rest of the initial data set for the conformal wave equations
\eqref{ReducedWaveCFE1}-\eqref{ReducedWaveCFE5} can be
computed. Indeed, a calculation shows that:
\begin{subequations}
\begin{eqnarray}
&& s =\frac{1}{3} \bigg( \Delta \Omega + \frac{1}{4}\Omega\big(r -
   K_{ij}K^{ij} +K^2  \big)-\Sigma K \bigg), \label{InitialData1}\\
&& L_{ij} = \frac{1}{\Omega}\bigg( -D_iD_j \Omega + \Sigma K_{ij} + s
   h_{ij} \bigg), \label{InitialData2}\\
&& L_i = \frac{1}{\Omega}\big( K_i{}^k D_k\Omega - D_i \Sigma  \big), \label{InitialData3}
  \\
&& d_{ij} = \frac{1}{\Omega}\bigg( -L_{ij} +l_{ij} +\big( K \big(
 K_{ij} -\displaystyle\frac{1}{4} K h_{ij}\big) - K_{ki}
  K_j{}^k + \displaystyle\frac{1}{4}  K_{kl} K^{kl}h_{\bmi\bmj}\big)
   \bigg), \label{InitialData4}\\
&& d_{ijk} = \frac{1}{\Omega}\big( D_j K_{ki}- D_k K_{ji} + h_{ik}L_j
   - h_{ij} L_k  \big). \label{InitialData5}
\end{eqnarray}
\end{subequations}
Observe that the above expressions are formally singular at the points
where $\Omega=0$. This observation leads to the following:

\begin{definition}[\textbf{anti-de Sitter-like initial data}]
\label{Definition:AdSData}
For an anti-de Sitter initial data set it is understood a 3-manifold $\mathcal{S}_\star$
with boundary $\partial \mathcal{S}_\star \approx \mathbb{S}^2$ together with
a collection of smooth fields $(\Omega,h_{ij},K_{ij},\Sigma)$ such that:
\begin{itemize}
\item[(i)] $\Omega>0$ on $\mbox{\em int}\, \mathcal{S}_\star$;
\item[(ii)] $\Omega=0$ and $|\mbox{\em
    d}\Omega|^2=\Sigma^2-\tfrac{1}{3}\lambda>0$ on
  $\partial\mathcal{S}_\star$; 
\item[(iii)] the fields $s$, $L_{ij}$, $L_i$, $d_{ij}$ and $d_{ijk}$
  computed from relations \eqref{InitialData1}-\eqref{InitialData5} extend smoothly to
  $\partial \mathcal{S}_\star$.
\end{itemize}
\end{definition}

\begin{remark}
{\em Anti-de Sitter-like initial data sets are closely related to so-called
hyperboloidal data sets for Minkowski-like spacetimes ---see
\cite{Kan96a}. By means of this correspondence it is possible to adapt
the existence results for hyperboloidal initial data sets in
\cite{AndChrFri92,AndChr94} to the anti-de Sitter-like setting. In
particular, this shows the existence of a large class of time
symmetric data ---i.e. data for which $K_{ij}=0$.}
\end{remark}

\begin{remark}
\label{VanishingZQInitialData}
{\em The fields given by equations \eqref{InitialData1}--\eqref{InitialData5}
represent part of the initial data required to evolve the system of wave
equations \eqref{ReducedWaveCFE1}--\eqref{ReducedWaveCFE5}. A calculation shows
that the remaining component, $n^a n^b L_{ab}$, can be computed directly from
$L_{ij}$ and the gauge function $\mathcal{R}(x)$. On the other hand, the normal
derivatives of the fields $s, \ L_{ab}$ and $d^a{}_{bcd}$ on
$\mathcal{S}_\star$ can be computed via the system \eqref{CFE1}--\eqref{CFE4}
along with the contracted Bianchi identity. Furthermore, notice that this construction
guarantees that the zero--quantities trivially vanish on $\mathcal{S}_\star$}.
\end{remark}

\section{General set-up}
\label{Section:GeneralSetUp}

In this section we discuss in detail the gauge fixing and the boundary
data prescription for an initial-boundary problem for the conformal
Einstein field equations which, in turn, gives rise to anti-de Sitter-like spacetimes.

\medskip
In what follows, let $(\mathcal{M},g_{ab},\Xi)$ denote  a conformal
extension of an anti-de Sitter-like spacetime
$(\tilde{\mathcal{M}},\tilde{g}_{ab})$ with $g_{ab}=\Xi^2 \tilde{g}_{ab}$. It
will be assumed that the spacetime is causal (i.e. it contains no
closed timelike curves) and that it contains a smooth, oriented and
compact spacelike hypersurface $\mathcal{S}_\star$ with boundary $\partial
 \mathcal{S}_\star$ which intersects the conformal boundary $\mathscr{I}$ in
 such a way that $\mathcal{S}_\star \cap \mathscr{I}= \partial
 \mathcal{S}_\star$. It is convenient to define $\tilde{\mathcal{S}_\star} \equiv
 \mathcal{S}_\star\setminus \partial \mathcal{S}_\star$. The portion of
$\mathscr{I}$ in the future of $\mathcal{S}_\star$ will be denoted by
 $\mathscr{I}^+$. Furthermore, it will be assumed that the causal
 future $J^+(\mathcal{S}_\star)$ coincides with the future domain of
dependence $D^+(\mathcal{S}_\star\cup \mathscr{I}^+)$ and that $\mathcal{S}_\star
 \cup \mathscr{I}^+\approx [0,1)\times \mathcal{S}_\star$ so that, in
  particular, $\mathscr{I}^+ \approx [0,1)\times \partial \mathcal{S}_\star$.

\begin{figure}
\begin{center}
\includegraphics[scale=1]{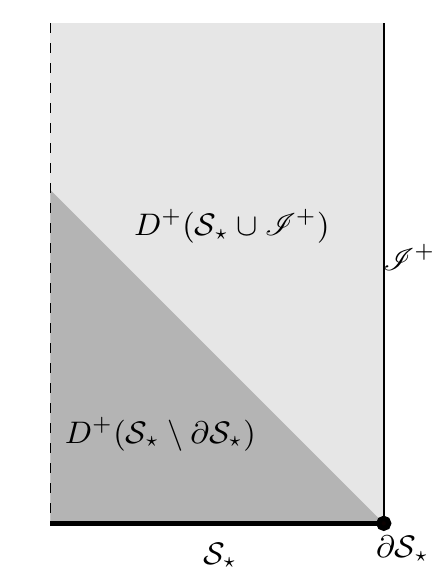}
\end{center}
\caption[Diagram of the construction of anti-de Sitter-like spacetimes.]
      {Penrose diagram of the set-up for the construction of anti-de
        Sitter-like spacetimes as described in the main text. Initial
        data prescribed on
        $\mathcal{S}_\star\setminus\partial\mathcal{S}_\star$ allows to
      recover the dark shaded region
      $D^+(\mathcal{S}_\star\setminus\partial\mathcal{S}_\star)$. In
      order to recover $D^+(\mathcal{S}_\star \cup \mathscr{I}^+)$ it
      is necessary to prescribe boundary data on
      $\mathscr{I}^+$. Notice that $D^+(\mathcal{S}_\star \cup \mathscr{I}^+)=J^+(\mathcal{S}_\star)$.}
    \label{Figure:AdSDomains}
\end{figure}

\subsection{Coordinates}
Close to the conformal boundary $\mathscr{I}$ we will make use of adapted
coordinates $x=(x^\mu)$ such that in terms of these coordinates
\[
\mathscr{I} =\{ x\in \mathbb{R}^3 \; | \;  x^1=0   \}.
\]
The coordinate $x^0$ is chosen so that the initial hypersurface
$\mathcal{S}_\star$ corresponds to the condition $x^0=0$. Accordingly,
the corner $\partial \mathcal{S}_\star$ is described by the conditions
$x^0=0$ and $x^1=0$. 

The coordinates $x=(x^\mu)$ are propagated off the initial
hypersurface $\mathcal{S}_\star$ through the generalised wave
coordinate condition
\begin{equation}
\square x^\mu = -\mathcal{F}^\mu(x).
\label{GeneralisedWaveCoordinateConditionALT}
\end{equation}
The value of the coordinates on $\mathcal{S}_\star$ provides the
initial data for the equation
\eqref{GeneralisedWaveCoordinateConditionALT}. The initial value of
the normal derivatives to $\mathcal{S}_\star$ is obtained from the
requirement that $(x^\mu)$ are independent ---that is, the coordinate
differentials $\mathbf{d}x^\mu$ must be linearly independent. 


\subsection{Boundary conditions for the conformal evolution equations}
\label{Section:BoundaryConditions}
In this subsection we discuss the boundary conditions to be imposed on
the various conformal fields. In \cite{Fri95} it has been shown that
it is possible to formulate an initial boundary-initial value problem
for anti-de Sitter-like spacetimes in which the conformal class of the
metric on the conformal boundary is specified freely. In the following, 
we investigate whether it is possible to make a similar prescription
in our scheme. More precisely, we would like to specify \emph{Dirichlet boundary data} for
the wave equations \eqref{ReducedWaveCFE1}-\eqref{ReducedWaveCFE5}
---that is, one would like to specify the values of the scalar fields
$\Xi$, $s$ and the components of the tensors $g_{\mu\nu}$,
$\Phi_{\mu\nu}$ and $d_{\mu\nu\lambda\rho}$ on $\mathscr{I}$.

\subsubsection{Boundary data for the conformal factor}
The evolution of the conformal factor $\Xi$ is described by the
wave equation \eqref{ReducedWaveCFE1}. For this equation
one naturally prescribes Dirichlet boundary conditions
such that
\[
\Xi \simeq 0.
\]
In other words, one has that $\Xi = O(x^1)$ close to $\mathscr{I}$. On
$\mathcal{S}_\star$ one wants to identify $\Xi$ with some
3-dimensional conformal factor $\Omega$ such that $\Omega=0$,
$\mathbf{d}\Omega\neq 0$ at $\partial\mathcal{S}_\star$, consistent
with Definition \ref{Definition:AdSData}.

\subsubsection{The Friedrich scalar}
The evolution of the Friedrich scalar $s$ is governed by the wave equation
\eqref{ReducedWaveCFE2}. In the context of the conformal constraint equations
on the conformal boundary, the Friedrich scalar $s$ is a gauge dependent
quantity which contains information about the manner the conformal boundary
embeds in the spacetime. Following Proposition
\ref{Proposition:ConformalConstraintsConformalBoundary} we set
\begin{equation}
s\simeq \varkappa(x) \notSigma, \qquad \notSigma =
\sqrt{\frac{|\lambda|}{3}}, \qquad \notK_{ij} \simeq \varkappa(x) \ell_{ij},
\label{BoundaryPrescription:SigmaFriedrich}
\end{equation}
where $\varkappa(x)$ is an arbitrary scalar field. This specification
of $s$ is independent of the choice of the gauge source function
$\mathcal{R}(x)$ associated to the Ricci scalar ---see the discussion
in Remark \ref{Remark:ConformalGaugeSourceFunction}. In particular, it
is possible, say, to have two related conformal representations of the
same physical solution with
the same spacetime Ricci scalar, one with a conformal boundary which is extrinsically curved and
the other extrinsically flat. 


\begin{remark}
{\em Observe that the particular choice $\varkappa(x)=0$ renders a
  conformal boundary which is \emph{extrinsically flat} with respect
  to the ambient spacetime ---see equation \eqref{SolutionConstraints1}. }
\end{remark}



\subsubsection{Boundary data for the components
  of the conformal
  metric}
  

\medskip
In the following it is convenient to make use of the 3+1 decomposition
of the metric $g_{ab}$ with respect to the unit normal to the conformal
boundary ---namely
\[
\bmg = \notalpha^2 \mathbf{d}x^1 \otimes \mathbf{d}x^1 + \ell_{\gamma\delta}
\big( \notbeta^\gamma \mathbf{d}x^1 + \mathbf{d}x^\gamma\big)\otimes
\big( \notbeta^\delta \mathbf{d}x^1 + \mathbf{d}x^\delta\big),
\qquad \gamma,\, \delta = 0,\,2,\,3.
\]
In particular, $(\ell_{\gamma\delta})$ denote the components of the intrinsic
metric $\ell_{ij}$ of the conformal boundary and $\notalpha$ and
$\notbeta^\gamma$ are, respectively, the lapse and shift. As $\mathscr{I}$ is
timelike, then $\ell_{ij}$ is a 3-dimensional Lorentzian metric of signature
$(-++)$. Accordingly, the components $(g_{\mu\nu})$ are given by
\begin{equation}
(g_{\mu\nu}) =
\left(
\begin{array}{cc}
\notalpha^2 +\notbeta_\gamma\notbeta^\gamma & \notbeta_\gamma\\
\notbeta_\delta & \ell_{\gamma\delta}
\end{array}
\right),
\label{DecompositionMetric:ConformalBoundary}
\end{equation}
so that for the components of the contravariant metric one has
\[
(g^{\mu\nu}) =
\left(
\begin{array}{cc}
\notalpha^{-2} & -\notalpha^{-2}\notbeta^\gamma\\
-\notalpha^{-2}\notbeta^\delta & \ell^{\gamma\delta}+\notalpha^{-2}\notbeta^\gamma\notbeta^\delta
\end{array}
\right).
\]

\begin{remark}
{\em In the following we regard the components
$(\ell_{\alpha\beta})$ as our basic boundary data.}
\end{remark}

Without loss of generality, we adopt a \emph{Gaussian gauge}
at the conformal boundary so that 
\begin{equation}
\notalpha\simeq 1, \qquad \notbeta^\gamma \simeq 0,
\label{MetricGaugeConditionConformalBoundary}
\end{equation}
and the metric $g_{ab}$ takes the form
\[
\bmg \simeq \mathbf{d}x^1 \otimes \mathbf{d}x^1 + \ell_{\alpha\beta}
\mathbf{d}x^\alpha \otimes \mathbf{d}x^\beta.
\]

\begin{remark}
{\em The prescription of the gauge conditions at the conformal boundary 
  \eqref{MetricGaugeConditionConformalBoundary} is independent of the
  generalised harmonic condition
  \eqref{GeneralisedWaveCoordinateConditionALT} and, thus, consistent
  with each other. Indeed, a calculation shows that for a metric in
  the form given by \eqref{DecompositionMetric:ConformalBoundary} one has that 
\begin{subequations}
\begin{eqnarray}
&& \hspace{-1cm} \Gamma^1 = \frac{1}{\notalpha^3}\big(\partial_1 \notalpha
   -\notbeta^\gamma \partial_\gamma \notalpha +\notalpha^2 \notK\big), \label{GeneralisedWaveCoordinatesDecomposed1}\\
&& \hspace{-1cm} \Gamma^\delta = \gamma^\delta -
   \frac{\notbeta^\delta}{\notalpha^3}\big( \partial_1 \notalpha
   -\notbeta^\gamma \partial_\gamma \notalpha + \notalpha^2 \notK \big) +
   \frac{1}{\notalpha^2}\big( \partial_1 \notbeta^\delta
   -\notbeta^\gamma \partial_\gamma \notbeta^\delta + \notalpha \partial^\delta
   \notalpha\big),
\label{GeneralisedWaveCoordinatesDecomposed2}
\end{eqnarray}
\end{subequations}
and $\gamma^\delta \equiv
\ell^{\eta\theta}\gamma^\delta{}_{\eta\theta}$ denote the
3-dimensional contracted Christoffel symbols. Thus, the  generalised harmonic condition
  \eqref{GeneralisedWaveCoordinateConditionALT} only prescribes the
  propagation of the gauge fields $\notalpha$ and $\notbeta^\gamma$ off the
  conformal boundary and do not constraint the components of the
  3-metric $\ell_{ij}$. Observe that $\notalpha$ and $\notbeta^\gamma$
  depend on the choice of $\varkappa(x)$ as $\notK=3\varkappa(x)$ as a
  consequence of equation
  \eqref{BoundaryPrescription:SigmaFriedrich}. 
}
\end{remark}

\subsubsection{Boundary data for the components
  of the Schouten tensor}

Given the 3-metric $\ell_{ij}$ of the conformal boundary, one can
compute the tangential components $(\notL_{\alpha\beta})$ and
tangential-normal components $(\notL_\alpha)$ of the spacetime
Schouten tensor at the conformal boundary using formulae
 \eqref{SolutionConstraints2} and \eqref{SolutionConstraints3}. One has then that
\begin{equation}
\notL_{\alpha}\simeq -\notD_\alpha \varkappa(x), \qquad \notL_{\alpha\beta}\simeq
\notl_{\alpha\beta} - \frac{1}{2}\varkappa^2(x) \ell_{\alpha\beta},
\label{BoundaryPrescription:Schouten}
\end{equation}
where $\varkappa(x)$ is the arbitrary scalar field determining the
extrinsic curvature of the conformal boundary according to equation
\eqref{SolutionConstraints1} and $\notl_{\alpha\beta}$ denotes the
components of the Schouten tensor $\notl_{ij}$ of the metric
$\ell_{ij}$. To compute the normal-normal component $\notL_{11}$ we notice
that
\[
 g^{\mu\nu}\notL_{\mu\nu}=\frac{1}{6}R.
\]
Thus, one has that 
\begin{eqnarray}
&& \notL_{11} \simeq \frac{1}{6}\mathcal{R}(x) - \ell^{\alpha\beta} \notl_{\alpha\beta} +
\frac{1}{2}\varkappa^2(x)
   \ell_{\alpha\beta}\ell^{\alpha\beta}\nonumber \\
&& \phantom{\notL_{11}}\simeq \frac{1}{6}\mathcal{R}(x) - \frac{1}{4}r
   +\frac{3}{2}\varkappa^2(x) \label{SchoutenNormalNormal}
\end{eqnarray}
where it is recalled that $\mathcal{R}(x)$ denotes the conformal gauge source
function introduced in Remark \ref{Remark:ConformalGaugeSourceFunction}. 

\subsubsection{The rescaled Weyl tensor}
The boundary data for the magnetic part of the rescaled Weyl tensor is
directly  computed from the metric $\ell_{ij}$ using the formula
\begin{equation}
\notd_{ijk}\simeq - \sqrt{\frac{3}{|\lambda|}}y_{ijk}, 
\label{BoundaryPrescription:MagneticWeyl}
\end{equation}
where $y_{ijk}$ denotes the Cotton tensor of $\ell_{ij}$---see
equation \eqref{SolutionConstraints4} in Proposition
\ref{Proposition:ConformalConstraintsConformalBoundary}. 

\medskip
The computation of the boundary data for the electric part requires
more work. From the discussion in Section
\ref{Section:ConformalConstraintsConformalBoundary} it follows that
the electric part of the rescaled Weyl tensor satisfies on $\mathscr{I}$ the equation
\begin{equation}
\notD^i \notd_{ij}\simeq 0.
\label{GaussConstraint}
\end{equation}
We now consider a $2+1$ decomposition of this equation on $\mathscr{I}$. To this end
let $\partial\mathcal{S}_t$, $t\in [0,\infty)$ with $\partial
\mathcal{S}_0=\partial\mathcal{S}_\star$ denote a foliation of the
conformal boundary and let $\nu_i$ denote the normal to this
foliation. The projector $s_i{}^j$ onto the leaves
$\partial\mathcal{S}_t$ is given by
\[
s_{ij} = \ell_{ij} + \nu_i \nu_j.
\]
The covariant derivative $\notD_i$ can be decomposed, in turn, as
\[
\notD_i = -\nu_i \delta + \delta_i
\]
where $\delta $ is the covariant directional derivative in the
direction of $\nu_i$ and $\delta_i$ is the Levi-Civita covariant
derivative associated to the 2-dimensional metric $s_{ij}$. The normal
$\nu_i$ induces the decomposition
\[
\notd_{ij} = w_{ij} -\nu_i w_j -\nu_jw_i + \nu_i \nu_j w,
\qquad w_{ij}=w_{(ij)}, 
\]
of the electric part of the rescaled Weyl tensor, where
\[
w_{ij} \equiv s_i{}^k s_j{}^l\notd_{kl}, \qquad w_i \equiv s_i{}^k
\nu^l \notd_{kl}, \qquad w\equiv \nu^i \nu^j \notd_{ij}.
\]
Using the above expressions, and observing that $w = w_i{}^i$, one
obtains the following decomposition of equation \eqref{GaussConstraint}:
\begin{subequations}
\begin{eqnarray}
&& \delta w  - \delta^i w_i = -\frac32 k w - k^{ij} w_{\{ij\}}, \label{GaussConstraintDecomposed1}\\
&& 2\delta w_i - \delta_iw = -2kw_i  - 2k_i{}^j w_j + 2\delta^j w_{\{ij\}}, \label{GaussConstraintDecomposed2}
\end{eqnarray}
\end{subequations}
where the 2-dimensional extrinsic curvature of the leaves of the foliation
$\partial\mathcal{S}_t$, $k_{ij}$, and the acceleration, $a_i$, are defined
via the relation
\[
\notD_i \nu_j = k_{ij} + \nu_i a_j, \qquad k\equiv s^{ij} k_{ij},
\]
and $ w_{\{ij\}} \equiv w_{ij}- \frac{1}{2}s_{ij} w$ is the
$s$-tracefree part of $w_{ij}$. 

\begin{remark}
{\em 
Expressing equations
\eqref{GaussConstraintDecomposed1}-\eqref{GaussConstraintDecomposed2}
in terms of coordinates $(x^{\mathcal{A}})=(t,x^{A})$ adapted to the
foliation $\partial \mathcal{S}_t$, one finds that the former imply a
first order symmetric hyperbolic system for $w$ and the two
non-trivial independent components $w_{\mathcal{A}}$ of $w_i$ provided
that the components $w_{\{AB\}}$ are known. Thus, the components 
$w_{\{AB\}}$ of the electric part of the rescaled Weyl tensor constitute an
independent piece of boundary data that supplements the prescription
of the Lorentzian 3-metric $\ell_{ij}$.}
\end{remark}

\begin{remark}
{\em The restriction to $\mathscr{I}$ of the generalised wave
coordinate conditions \eqref{GeneralisedWaveCoordinateConditionALT}
allows to specify, via the relation
\eqref{GeneralisedWaveCoordinatesDecomposed2}, a natural choice for the
lapse and shift (and thus a choice of the foliation of
$\partial\mathcal{S}_t$) for which equations
\eqref{GaussConstraintDecomposed1}-\eqref{GaussConstraintDecomposed2}
are to be solved. }
\end{remark}

The discussion of the previous paragraphs leads to the following:

\begin{lemma}
\label{Lemma:CompletingElectricPartConformalBoundary}
Let on $\mathscr{I}$ be given:
\begin{itemize}

\item[(i)] a smooth 3-dimensional Lorentzian metric $\ell_{ij}$; 

\item[(ii)] a prescription of coordinate gauge source functions
$\mathcal{F}^{\mu}(x)$ and the intrinsic gauge function
$\varkappa(x)$; 

\item[(iii)] a smooth symmetric tensor
$w_{\{ij\}}$ which is spatial with respect to the
foliation induced on $\mathscr{I}$ by  the functions
$\mathcal{F}^{\mu}(x)$ and tracefree with respect to the metric
induced on the leaves of the foliation;

\item[(iv)] a smooth choice of fields $w$ and $w_i$ on a fiduciary
  hypersurface $\partial\mathcal{S}_\star$ of $\mathscr{I}$.
\end{itemize}
Then, there exists a $t_\bullet>0$ such that on
$\mathscr{I}_{t_\bullet}\approx
[0,t_\bullet)\times \partial\mathcal{S}_\star $ there exists unique
fields $w$ and $w_i$ which together with the prescribed choice of
$w_{\{ij\}}$ satisfy the constraint \eqref{GaussConstraint}.
\end{lemma}

\begin{proof}
The proof of this result follows from the discussion in the previous
paragraphs and the theory of local existence of first order symmetric
hyperbolic systems.
\end{proof}

\begin{remark}
{\em The \emph{free data} $w_{\{ij\}}$ can be related to the notion of
incoming and outgoing radiation. In order make this evident, let
$(\bml, \bml', \bmm, \bm{\bar{m}})$ be a Newman--Penrose tetrad satisfying
the following normalisation relations in accordance with our conventions:

\[
l_a l'^a = -1, \quad m_a \bar{m}^a = 1,
\]

\noindent while all the remaining contractions vanish. The normal unit vectors
are expressed, respectively, as $n_a = \frac{1}{\sqrt{2}}(l_a + l'_a)$ and $\notn_a =
\frac{1}{\sqrt{2}}(l_a - l'_a)$. Using this, the different metrics take the
following form:
\[
g_{ab} = -2l_{(a} l'{}_{b)} + 2m_{(a}\bar{m}_{b)}, \quad
\ell_{ab} = 2m_{(a}\bar{m}_{b)} - l_{(a}l'{}_{b)} -\tfrac12(l_al_b + l'{}_al'{}_b),
\quad s_{ab} = 2m_{(a}\bar{m}_{b)}.
\]
Making use of this and observing that $\omega_{ab} = \notn^q \notn^s s_a{}^p
s_b{}^r d_{dpqrs}$, an expansion of the Weyl tensor in terms of the tetrad
defined above  leads, after a straightforward calculation,
to:
\[
\omega_{\{ab\}} = \frac12 \bigg( (\psi_0 + \psi^*_4)\bar{m}_a\bar{m}_b +
(\psi^*_0 + \psi_4)m_a m_b \bigg),
\]
where $\psi_0 \equiv d_{pqrs}l^pm^ql^rm^s$ and $\psi_4 \equiv
d_{pqrs}l'^p\bar{m}^q l'^r\bar{m}^s$ ---see e.g. \cite{Cha98}.  This shows that $\psi_0$ and
$\psi_4$ constitute part of the basic data one must provide on $\mathscr{I}$.}

\end{remark}

\subsubsection{Summary}
The analysis of this section can be summarised as follows:

\begin{proposition}
\label{Proposition:SummaryBoundaryData}
Let on $\mathscr{I}$ be given a smooth Lorentzian metric $\ell_{ij}$
and a smooth tensor field $w_{\{ij\}}$ as in Lemma
\ref{Lemma:CompletingElectricPartConformalBoundary}. Moreover, let the
fields 
\[
\notSigma, \qquad s, \qquad \notK_{ij}, \qquad \notL_i, \qquad
\notL_{ij}, \qquad \notd_{ijk}
\]
be constructed according to formulae
\eqref{BoundaryPrescription:SigmaFriedrich}, \eqref{BoundaryPrescription:Schouten}
and \eqref{BoundaryPrescription:MagneticWeyl}.  Finally, 
let $\Theta_a$, $\Upsilon_a$, $\Delta_{abc}$ and $\Lambda_{abc}$ be
the zero-quantities defined by relations
\eqref{SubsidiaryDefinition1}-\eqref{SubsidiaryDefinition4}. One has then that 
\begin{eqnarray*}
&\ell_b{}^a\Theta_a\simeq 0, &\\
& \ell_c{}^a \ell_d{}^b \Upsilon_{ab}\simeq 0, \quad
\notn^a \ell_c{}^b\Upsilon_{ab}\simeq0, & \\
& \ell_e{}^c \ell_f{}^d
\ell_g{}^b \Delta_{cdb}\simeq 0, \quad \notn^b \ell_e{}^c \ell_f{}^d \Delta_{cdb}\simeq 0,& \\
& \notn^b \ell_e{}^c \ell_f{}^d\Lambda_{bcd}\simeq 0, \quad
  \notn^b \notn^d \ell_e{}^c \Lambda_{bcd}\simeq 0, & 
\end{eqnarray*}
at least on $\mathscr{I}_{t_\bullet}\approx
[0,t_\bullet)\times \partial\mathcal{S}_\star $, 
where $\notn^a$ and $\ell_a{}^b$ denote, respectively, the normal and
projector of the conformal boundary $\mathscr{I}$. 
\end{proposition}



\subsection{Corner conditions}
\label{Section:CornerConditions}

In the previous sections we have discussed the problem of the determination of
initial and boundary data. In particular, it is clear that once boundary data
have been provided on $\mathscr{I}$, time derivatives of the various conformal
fields can be directly calculated.  However, these data do not necessarily
match smoothly with the ones corresponding to $\mathcal{S}_\star$ at the
corner. The purpose of this section is to analyse the compatibility conditions,
at different orders, arising from the conformal Einstein field equations and
the wave equations ---these conditions are commonly known as \emph{corner
conditions}.  In the following, the subscript $_\odot$ will stand for a
quantity evaluated at $\partial\mathcal{S}_\star$. 

\subsubsection{Conditions for the metric}

In terms of the adapted coordinates previously introduced, the corner
$\partial\mathcal{S}_\star$ is defined by the conditions $x^0 = 0$ and  $x^1=0$. Exploiting the
gauge freedom, we adopt local Gaussian coordinates both on $\mathcal{S}_\star$ and
$\mathscr{I}$. Denoting as $h_{\gamma\delta}$ and
$\ell_{\mathcal{A}\mathcal{B}}$ the intrinsic 3--metrics corresponding to these
hypersurfaces, respectively, this condition implies that the spacetime metric at
$\partial\mathcal{S}_\star$ can be written in the two following ways:
\begin{eqnarray*}
&& \bmg = -\mathbf{d}x^0 \otimes \mathbf{d}x^0 + h_{\gamma\delta}
\mathbf{d}x^\gamma \otimes \mathbf{d}x^\delta, \qquad (\gamma, \delta = 1,2,3),\\
&& \bmg = \mathbf{d}x^1 \otimes \mathbf{d}x^1 + \ell_{\mathcal{A}\mathcal{B}}
\mathbf{d}x^\mathcal{A} \otimes \mathbf{d}x^\mathcal{B}, \qquad (\mathcal{A}, \mathcal{B} = 0,2,3).
\end{eqnarray*}
Hereafter, the previous convention for the indices will be used.  Additionally,
uppercase indices $A,\, B,\ldots$ will stand for the coordinates $x^2$ and
$x^3$ (which we will refer to as \emph{angular}) of the sections of
$\mathscr{I}$.

\medskip
\noindent
\textbf{Zero order conditions.} Comparing the two last expressions for the
metric, one readily finds that
\begin{equation}
(\ell_{00})_\odot = -1, \quad  (h_{11})_\odot = 1, \quad
(\ell_{AB})_\odot = (h_{AB})_\odot,
\label{zeroconds}
\end{equation}
while the remaining components vanish at $\partial \mathcal{S}_\star$.

\medskip
\noindent
\textbf{First order conditions.} In Gaussian coordinates, we can express the
normal derivatives of the metric in terms of the corresponding extrinsic
curvature. Explicitly, one has:
\begin{subequations}
\begin{eqnarray} 
&& K_{\gamma\delta}|_{_{\mathcal{S}_\star}} = \frac12\partial_0 h_{\gamma\delta}|_{_{\mathcal{S}_\star}}
= \Gamma^0{}_{\gamma\delta}|_{_{\mathcal{S}_\star}}, \label{curvature1}\\
&& \notK_{\mathcal{A}\mathcal{B}} \simeq \frac12\partial_1\ell_{\mathcal{A}\mathcal{B}} 
\simeq -\Gamma^1{}_{\mathcal{A}\mathcal{B}}.
\label{curvature2}
\end{eqnarray}
\end{subequations}
As $K_{\gamma\delta}$ is part of the initial data, this establishes a corner
condition for $\partial_0 h_{\gamma\delta}$; in particular, the angular
components must satisfy the condition $(\partial_0 h_{AB})_\odot = (\partial_0
\ell_{AB})_\odot$.

Recall that in Gaussian coordinates the propagation of the timelike vector
$(\partial_0)^a$ along itself implies that $\Gamma^\mu_{00}|_{_{\mathcal{S}_\star}}
=0$; similarly, for the normal to $\mathscr{I}$ one has that $\Gamma^\mu_{11}
\simeq 0$.  The previous conditions on the Christoffel symbols, along with
equations \eqref{curvature1} and \eqref{curvature2}, imply that $K_{11}$ and
$\notK_{00}$ vanish at the corner. Furthermore, the traces of the extrinsic
curvature can be related to the gauge functions $\mathcal{F}^\mu(x)$ as follows:
\[
K_\odot = (h^{AB}K_{AB})_\odot = \mathcal{F}^0(x)_\odot, \qquad
\notK_\odot = (\ell^{AB}\notK_{AB})_\odot = - \mathcal{F}^1(x)_\odot.
\]

\medskip

\noindent Finally, given that $\nabla$ is a Levi--Civita connection and the
acceleration is zero, our coordinate choice determines the remaining partial
derivatives: $(\partial_0 g_{0\mu})_\odot = -(\Gamma^0_{0\mu})_\odot = 0$.

\medskip
\noindent
\textbf{Second order conditions.} Second order conditions can be
extracted in a straightforward way from the wave equation for the 
metric, equation \eqref{ReducedWaveCFE5} ---namely
\begin{equation*}
g^{\lambda\rho}\partial_\lambda\partial_\rho g_{\mu\nu}= 2\bigg( g_{\lambda\rho}
  g^{\sigma\tau} \Gamma^\lambda{}_{\sigma\mu} \Gamma^\rho{}_{\tau\nu}
  + 2 \Gamma^\sigma{}_{\lambda\rho} g^{\lambda\tau} g_{\sigma(\mu} \Gamma^\rho{}_{\nu)\tau} 
- g_{\sigma(\mu}\nabla_{\nu)} \mathcal{F}^\sigma(x) - 2\Phi_{\mu\nu} - \frac14 g_{\mu\nu} \mathcal{R}(x)
\bigg).
\end{equation*}
Using the conditions discussed above for the first order derivatives, the wave
equation for the components $g_{\mu\nu}$ can be written schematically as:
\begin{eqnarray*}
&& (\partial^2_0 \ell_{\mu\nu})_\odot = (\partial^2_1 h_{\mu\nu})_\odot +
(h^{CD}\partial_C\partial_Dh_{\mu\nu})_\odot + f_{\mu\nu}(\bmg, \bmK, \bm{\notK}, 
\bm{\mathcal{F}}(x), \bm\Phi, \mathcal{R}(x))_\odot.
\end{eqnarray*}
Apart from the components of the Schouten tensor encoded into $\Phi_{\mu\nu}$ (to be
discussed below), the second order condition can be expressed in terms of the
initial data, lower order corner conditions and gauge functions at the corner. Further
application of $\partial_0$ enables to obtain higher order conditions.

\subsubsection{Conditions for the conformal factor}

As, by definition, $\Xi =0$ on the conformal boundary, then all its intrinsic
derivatives of any order will vanish. In particular,
$\partial\mathcal{S}_\star$ automatically inherits these conditions.  Regarding
the normal derivative, solution \eqref{SolutionConstraints5} gives its value on
$\mathscr{I}$. Accordingly, one has that
\[
(\notSigma)_\odot = \sqrt{-\frac{\lambda}{3}}.
\]
When smoothness is imposed, higher order partial derivatives both on
$\mathcal{S}_\star$ as well as on $\mathscr{I}$ are forced to coincide at
$\partial\mathcal{S}_\star$.

\subsubsection{Conditions for the Friedrich scalar}

\textbf{Zero order condition.} As discussed previously, the Friedrich scalar
$s$ is determined on the conformal boundary by the gauge function
$\varkappa(x)$. Nevertheless, when the $00$ component of equation \eqref{CFE1}
is evaluated at the corner, our choice of Gaussian coordinates imply that,
\[
s_\odot = 0.
\]

\noindent \textbf{First order conditions.} Equation \eqref{CCCB3} ---or
alternatively \eqref{CFE2}--- determines the intrinsic derivatives of $s$ on the
boundary. In particular, the time derivative takes the following form at the
corner:
\[
(\partial_0 s)_\odot = -\notSigma (L_{01})_\odot.
\]
This expression is equivalent to the one given in
\eqref{BoundaryPrescription:Schouten} for the tangential-normal components of
$L_{ab}$ on $\mathscr{I}$. 

\medskip

\noindent \textbf{Second order conditions.} The second order condition for $s$
can be extracted from the wave equation \eqref{ReducedWaveCFE2} expressed in
Gaussian coordinates. The evaluation of this equation at the corner yields:
\[
(\partial^2_0 s)_\odot = (\partial^2_1 s)_\odot + (h^{AB}\partial_A\partial_B s)_\odot - 
(\mathcal{F}^\mu(x)\partial_\mu s 
+ \frac16(s\mathcal{R}(x) + \notSigma\partial_1 \mathcal{R}(x))_\odot.
\]
\noindent Here, the spatial derivatives of $s$ can be computed from the
restriction of the initial data to $\partial\mathcal{S}_\star$ while $\partial_0 s$
corresponds to the first order condition. The functions $\mathcal{F}^\mu (x)$
and $\mathcal{R}(x)$ are gauge-dependent prescribed quantities. Furthermore, we
observe that $\partial^2_0 s$ is written in terms of the first order
derivatives, indicating then a recursive procedure to find higher order
conditions ---computed by further application of $\partial_0$ to equation
\eqref{ReducedWaveCFE2}.

\subsubsection{Conditions for the Schouten tensor}

Next, we will show how the constraint equations impose restrictions on
the components of $L_{ab}$, which along with the gauge quantity
$\mathcal{R}(x)$ determines the tracefree tensor $\Phi_{ab}$ on
$\mathscr{I}$.

\medskip

\noindent \textbf{Zero order corner conditions.} The value of components
$L_{\alpha\beta}$ and $L_{0\alpha}$ at the corner can be obtained from the
initial data \eqref{InitialData2} and \eqref{InitialData3} taking the limit
$\Omega \to 0$.  Imposing smoothness, they must match the boundary data given
by equations \eqref{BoundaryPrescription:Schouten} and
\eqref{SchoutenNormalNormal} at $\partial\mathcal{S}_\star$. The same is imposed for
component $L_{00}$. 


\medskip

\noindent \textbf{First order corner conditions.} 
First time derivatives of the components $L_{\alpha\beta}$ and $L_{0\alpha}$
can be obtained via equation \eqref{CFE3}. More explicitly one has:
\begin{eqnarray*}
&& (\partial_0 L_{\alpha\beta})_\odot = \notSigma (d^1{}_{\beta 0 \alpha})_\odot + 
f_{\alpha\beta}(\bmL, \bmh, \bmK, \bm{\notK})_\odot, \\
&& (\partial_0 L_{\alpha_0})_\odot = \notSigma (d^1{}_{00\alpha})_\odot +
f_{\alpha}(\bmL, \bmh, \bmK, \bm{\notK})_\odot.
\end{eqnarray*}
As it will be seen below, the components of the Weyl tensor appearing here, are
part of the data satisfying zero-order conditions, so they must be consistent
with the last equations. On the other hand, a condition for $(\partial_0
L_{00})_\odot$ can be obtained via the contracted Bianchi identity.

\medskip

\noindent \textbf{Second order corner conditions.} Second order time
derivatives of $L_{ab}$ are to be obtained by evaluating the wave equation
\eqref{ReducedWaveCFE3} at $\partial\mathcal{S}_\star$. For $L_{\alpha\beta}$ one
has:
\[
(\partial_0^2 L_{\alpha\beta})_\odot = (\partial^2_1 L_{\alpha\beta})_\odot 
+ (h^{CD}\partial_C\partial_D L_{\alpha\beta})_\odot 
+ f_{\alpha\beta}(\bmh, \bmL, \bmK, \bm{\notK}, \partial\mathcal{F}(x),
\mathcal{R}(x))_\odot.
\]
Similar expressions can be obtained for the rest of the components.

\subsubsection{Conditions for the Weyl tensor}

Information about the Weyl tensor is encoded in the electric and magnetic
parts. These are given on $\mathcal{S}_\star$ by equations \eqref{InitialData4} and
\eqref{InitialData5}, and has been discussed in section 4.2.5 for
$\mathscr{I}$.  As these data have been obtained using different projections,
their components must be carefully matched. One can check that they
share the components $d_{0101}, \ d_{010A}, \ d_{01A1}, \ d_{01AB}$ and
$d_{0A1B}$ so, when matched, they represent the zero-order conditions.

\medskip

\noindent \textbf{First order corner conditions.} Given the structure of
equation \eqref{CFE4}, only certain conditions can be extracted from it.
Ultimately, when it is evaluated at the corner it takes the form:

\[
(\partial_0 d^0{}_{\lambda\mu\nu})_\odot = f_{\lambda\mu\nu} (\bmK, \bm{\notK}, \bmd)_\odot.
\]

\medskip

\noindent \textbf{Second order corner conditions.} Second order time
derivatives of the rescaled Weyl tensor are given by the wave equation
\eqref{ReducedWaveCFE4}. As $\Xi$ vanishes at the corner, the equation is
significantly simplified. Expanding the reduced wave operator $\blacksquare$
it takes the schematic form

\[
(\partial^2_0d_{\lambda\mu\nu\sigma})_\odot = (\partial^2_1 d_{\lambda\mu\nu\sigma})_\odot
+ (\partial_A\partial_B d_{\mu\nu\lambda\sigma})_\odot + 
f_{\lambda\mu\nu\sigma}(\bmg, \bmK, \bm{\notK}, \bmd)_\odot.
\]

\subsubsection{Concluding remarks regarding the corner conditions}
The discussion in the previous paragraphs provides a recursive
procedure to compute the corner conditions to any required
order. Given this procedure, its natural to ask whether there exist
any examples of pairs of initial data and boundary conditions which
satisfy the corner conditions to \emph{any arbitrary order}. The
difficulties in implementing corner conditions to any arbitrary order
have been discussed in \cite{Fri14}. A way of satisfying corner
conditions to an arbitrary order is to make use of the gluing
constructions for asymptotically hyperbolic initial data sets in
\cite{ChrDel09}. Given an asymptotically hyperbolic initial data set
satisfying certain smallness conditions, these constructions allow to
deform the data by a deformation which is supported arbitrarily far in
the asymptotic region, to ones which are exactly Schwarzschild-anti de
Sitter in the asymptotic region. This class of data is naturally
supplemented by Schwarzschild-anti de
Sitter boundary initial data ---and thus it trivially satisfies the
corner conditions to any order. The resulting spacetime has, accordingly, a
very special behaviour near the corner. In particular, the metric
$\ell_{ij}$ must be conformally flat near the corner. It is of
interest to analyse whether it is possible to construct a more
general class of initial--boundary data for adS-like spacetimes
satisfying the corner conditions at any order.  

\section{Propagation of the constraints}
\label{Section:PropagationOfTheConstraints}

The purpose of this section is to analyse the propagation of the gauge
conditions and to discuss the relation of the evolution system
\eqref{ReducedWaveCFE1}-\eqref{ReducedWaveCFE5} to the
Einstein field equations.

\subsection{Boundary conditions for the subsidiary equations}
The purpose of this section is to show that the boundary conditions for the
conformal wave equations \eqref{ReducedWaveCFE1}-\eqref{ReducedWaveCFE5}
discussed in the previous section imply trivial (i.e. vanishing) Dirichlet
boundary conditions for the subsidiary wave equations
\eqref{SubsidiaryEquation1}-\eqref{SubsidiaryEquation4}. 

\subsubsection{Transport equations for the subsidiary fields}

Proposition \ref{Proposition:SummaryBoundaryData} shows that as a consequence
of our Dirichlet boundary data prescription, the components of the zero fields
$\Theta_a$, $\Upsilon_a$, $\Delta_{abc}$ and $\Lambda_{abc}$ which only involve
derivatives intrinsic to $\mathscr{I}$ vanish. In order to show that the
remaining components also vanish, it is necessary to construct suitable
transport equations for the zero-quantities on the conformal boundary.  As it
will be seen, the equations for $\Upsilon_{ab}$ and $\Theta_a$ can be
constructed in a straightforward manner, whereas $\Delta_{abc}$ and
$\Lambda_{abc}$ require a more detailed treatment. The integrability conditions
for the zero--quantities \eqref{IC1}--\eqref{IC4} ---see Appendix A.2--- will
prove to be key to obtain these equations. In what follows let $\tau^i$ denote
a timelike vector on $\mathscr{I}$ with pushforward to the spacetime
$(\mathcal{M},g_{ab})$ given by $\tau^a$ and let $\mathcal{P}\equiv
\tau^a\nabla_a$. Notice then that $\notn^a\tau_a = 0$

\medskip
\noindent
\textbf{Transport equations for $\Theta_a$ and $\Upsilon_{ab}$.}
First, consider the expression $2\tau^a \nabla_{[a}\Upsilon_{b]c}$. On
the one hand, a calculation shows that
\[
2\tau^a\nabla_{[a}\Upsilon_{b]c} = \mathcal{P}\Upsilon_{bc} + \Upsilon_{ac}\chi_b{}^a
-\nabla_b(\tau^a\Upsilon_{ac}) \simeq  \mathcal{P}\Upsilon_{bc} + \Upsilon_{ac}\rchi_b{}^a,
\]
where $\rchi_{ab} \equiv \nabla_a\tau_b$ and the second equality follows from
the fact that $\ell_a{}^c \ell_b{}^d \Upsilon_{cd} \simeq 0$ and $\ell_a{}^c
n^b \Upsilon_{bc} \simeq 0$, which are a consequence of the validity of
constraints \eqref{CCCB1} and \eqref{CCCB2} on $\mathscr{I}$. On the other
hand, using the integrability condition \eqref{IC1} one
obtains the following transport equation:
\begin{equation}
\mathcal{P}\Upsilon_{bc} \simeq 2\tau^a g_{c[a}\Theta_{b]} - \Upsilon_{ac}\rchi_b{}^a,
\label{TE1}
\end{equation}
which crucially is homogeneous in the zero-quantities.

\smallskip
Now, for $\Theta_a$, consider the expression $2\tau^a\nabla_{[a}\Theta_{b]}$.
Expanding as in the case for $\Upsilon_{ab}$ one finds that 
\begin{eqnarray*}
&& 2\tau^a\nabla_{[a}\Theta_{b]} \simeq \mathcal{P} \Theta_b
   -\nabla_b\big(\tau^a\Theta_a\big) + \Theta_a \chi_b{}^a = \mathcal{P} \Theta_b
  + \Theta_a \rchi_b{}^a,
\end{eqnarray*}
where it has been used that $\ell_b{}^a\Theta_a\simeq 0$ ---as this is
equivalent to satisfy the constraint \eqref{CCCB3}--- so that
$\tau^a\Theta_a\simeq 0$. Using the integrability condition \eqref{IC2}, the
following homogeneous transport equation is directly obtained:
\begin{equation}
\mathcal{P} \Theta_b \simeq \tau^a\Delta_{abc}\nabla^c\Xi-
\tau^aL^c{}_{[a}\Upsilon_{b]c} - \Theta_a \rchi_b{}^a. \label{TE2}
\end{equation}

\medskip
\noindent
\textbf{Transport equations for $\Delta_{abc}$ and $\Lambda_{abc}$.}
For the zero-quantity $\Delta_{abc}$ consider
$3\tau^e\nabla_{[e}\Delta_{ab]c}$. A direct calculation shows that
\begin{equation}
3\tau^e\nabla_{[e}\Delta_{ab]c} = \mathcal{P}\Delta_{abc} - 2\rchi_{[a}{}^e
\Delta_{b]ec} + 2 \nabla_{[a}(\tau^e \Delta_{b]ec}). \label{TransportEquationDelta}
\end{equation}
As before, one needs to show that the last term in the previous expression
vanishes on the boundary. For this purpose a decomposition with respect to
$\ell_a{}^b$ can be performed.  Observing that the components $\ell_a{}^d
\ell_b{}^e \ell_c{}^f \Delta_{def} \equiv \Delta^{(3)}_{abc}$ and $\ell_a{}^c
\ell_b{}^d n^e \Delta_{cde}$ vanish by virtue of the constraints \eqref{CCCB4}
and \eqref{CCCB5}, as well as exploiting the fact that $\Delta_{[abc]} = 0$, a
calculation leads to
\[ \tau^b \Delta_{abc}
\simeq  \tau^b n_c \ell_a{}^d \ell_b{}^f n^e \Delta_{def} \equiv \tau^b n_c \Delta_{ab}.
\]
In view of this, it is enough to construct a further homogeneous transport
equation for $\Delta_{ab}$ on $\mathscr{I}$.  Performing suitable projections
in equation \eqref{TransportEquationDelta}, its right--hand side takes the
following form on the conformal boundary:
\[
\mathcal{P}\Delta_{ab} - \Delta_{def}\tau^c\nabla_c(\ell_a{}^d\ell_b{}^fn^e)
- n^e\ell_a{}^d\ell_b{}^f \rchi_e{}^c \Delta_{dcf} - \notD_a(\tau^e \Delta_{eb}),
\]
where the second term can be expressed in terms of the extrinsic curvature
$\notK_{ab}$. Then, using the integrability condition \eqref{IC3}, a homogeneous
transport equation for $\Delta_{ab}$ on the boundary is obtained.



\medskip

Finally, for $\Lambda_{abc}$ consider the expression $2\tau^a
\nabla_{[e}\Lambda_{a]bc}$:
\begin{equation}
2\tau^d \nabla_{[d}\Lambda_{a]bc} = \mathcal{P}\Lambda_{abc}
+\Lambda_{dbc}\rchi_a{}^d -\nabla_a(\tau^d\Lambda_{dbc}) = 0.
\label{TransportEquationLambda}
\end{equation}
where the last equality is consequence of the integrability condition
\eqref{IC4}.  When a decomposition is performed for $\tau^d\Lambda_{dbc}$, the
components $\ell_a{}^e \ell_b{}^f n^d \Lambda_{def}$ and $\ell_a{}^f n^d n^e
\Lambda_{def}$ vanish on $\mathscr{I}$ due to constraints
\eqref{CCCB6} and \eqref{CCCB7}. Then, a calculation yields:
\[
\tau^d\Lambda_{dbc} \simeq \tau^d \ell_d{}^e \ell_b{}^f \ell_c{}^g\Lambda_{efg}
- 2\tau^d n^e \ell_a{}^f \ell_c{}^g \Lambda_{fe[b} n_{g]} \equiv \tau^d\Lambda^{(3)}_{abc} 
- 2\tau^d\Lambda_{a[b}n_{c]}.
\]

As in the case for $\Delta_{abc}$, this means that suitable transport equations
must be constructed for $\Lambda^{(3)}_{abc}$ and $\Lambda_{ab}$. Making
suitable projections in equation \eqref{TransportEquationLambda} results in the
following transport equations:

\begin{subequations}
\begin{eqnarray*}
&& \mathcal{P}\Lambda^{(3)}_{abc} \simeq \Lambda_{efg}\tau^d\nabla_d(\ell_a{}^e\ell_b{}^f\ell_c{}^g)-
\ell_a{}^e\ell_b{}^f\ell_c{}^g\Lambda_{dfg}\rchi_e{}^d + \notD_a(\tau^d\Lambda^{(3)}_{dbc}), \\
&& \mathcal{P}\Lambda_{ab} \simeq \tau^d\nabla_d(\ell_a{}^f\ell_b{}^g n^e)\Lambda_{dfg}\rchi_e{}^d
+ \notD_a(\tau^d\Lambda_{db}) - \ell_b{}^g\tau^d\notD_f n^e \Lambda_{deg}.
\end{eqnarray*}
\end{subequations}

\begin{remark}
{\em The main observation following the previous calculations is that
  one has homogeneous propagation equations intrinsic to $\mathscr{I}$ for all the components of
  the zero-quantities which do not directly vanish by virtue of
  Proposition \ref{Proposition:SummaryBoundaryData}. Thus, if one can
  ensure that these intrinsic propagation equations have vanishing
  initial data at the corner, their solutions have to vanish along the
conformal boundary as well ---accordingly, the full set of
zero-quantities associated to the conformal field equations will
vanish on $\mathscr{I}$.}
\end{remark}

\subsubsection{The propagation argument}

Once we have obtained the relevant transport equations we are in position to
state the following lemma:

\begin{lemma}
\label{SubsidiaryEquations}
Consider vanishing initial data for the zero-quantities $\Upsilon_{ab}, \
\Theta_a, \ \Delta_{abc}$ and $\Lambda_{abc}$ at $\partial\mathcal{S}_\star$ and
assume that the conformal constraints \eqref{CCCB1}--\eqref{CCCB7} are
satisfied. Then, all the components of the zero-quantities vanish on the
conformal boundary.
\end{lemma}

\begin{remark}
\label{VanishingDerivativesZQInitialData}
{\em A similar approach can be employed to prove that the normal derivatives of
the zero--quantities vanish on $\mathcal{S}_\star$ via projecting the
integrability conditions \eqref{IC1}--\eqref{IC4} with respect to $n^a$. The
result readily follows from the fact that all the components of the zero--quantities
vanish on $\mathcal{S}_\star$ ---see \Cref{VanishingZQInitialData}. Thus, one
has vanishing initial data for the wave equations
\eqref{SubsidiaryEquation1}--\eqref{SubsidiaryEquation4}.} 
\end{remark}

\subsection{Propagation of the gauge}
The discussion of the propagation of the zero-quantities associated to
the conformal Einstein field equations needs to be supplemented with a
discussion of the propagation of the gauge. The strategy in this
regard is similar to that used in the analysis of the propagation of
the constraints ---i.e. one introduces a set of zero-quantities
associated to the gauge and purports constructing a suitable system of
subsidiary homogeneous evolution equations. 

\subsubsection{Basic relations}
In what follows it is convenient to define
\begin{subequations}
\begin{eqnarray}
&& Q \equiv R -\mathcal{R}(x), \label{QScalar} \\
&& Q^\mu \equiv \Gamma^\mu - \mathcal{F}^\mu(x), \label{QVector} \\
&& Q_{\mu\nu} \equiv R_{\mu\nu} -2 \Phi_{\mu\nu} - \frac{1}{4}\mathcal{R}(x)
g_{\mu\nu}. \label{QTensor}
\end{eqnarray}
\end{subequations}

\begin{remark}
{\em The zero-quantity $Q$ encodes the relation between the Ricci
  scalar of the unphysical spacetime and the conformal gauge source
  function. The zero-quantity $Q^\mu$ corresponds to the relation
  between the contracted Christoffel symbols and the coordinate gauge
  source function giving rise to the generalised wave
  coordinates. Finally, $Q_{\mu}$ is associated to the relation
  between the Ricci tensor and the reduced Ricci tensor
  ---compare with equation \eqref{UnphysicalEinsteinEquation}.}
\end{remark}

\begin{remark}
{\em In what follows we regard the field $g_{\mu\nu}$ as the components of
a metric tensor $g_{ab}$ in the coordinates $x=(x^\mu)$. Let
$R_{\mu\nu}$ denote the components of the Ricci tensor, $R_{ab}$, of
  $g_{ab}$ in the coordinates $(x^\mu)$ and let $R$ be the associated
  Ricci scalar. The objective of the subsequent analysis is to investigate
  under what circumstances one has that $R$ coincides with $\mathcal{R}(x)$,
  $R_{\mu\nu}$ coincides with $\mathscr{R}_{\mu\nu}$ and $\Phi_{\mu\nu}$ are the
  components of the
  symmetric tracefree part of $R_{\mu\nu}$ so that one can write
\[
R_{\mu\nu} = 2\Phi_{\mu\nu} +\frac{1}{4}\mathcal{R}(x) g_{\mu\nu}.
\]
This is equivalent to showing that
\[
Q=0, \qquad Q^\mu=0, \qquad Q^{\mu\nu}=0.
\]}
\end{remark}

The definitions of $Q^\mu$ and $Q_{\mu\nu}$ allows one to rewrite the
reduced Ricci operator, equation \eqref{DefinitionReducedRicci}, and
the reduced wave operator acting on $\Phi_{\mu\nu}$, equation
\eqref{ReducedWaveOperator}, as
\begin{subequations}
\begin{eqnarray}
&& \mathscr{R}_{\mu\nu}[\bmg] = R_{\mu\nu} -\nabla_{(\mu} Q_{\nu)}, \label{Identity:ReducedRicci}\\
&& \blacksquare \Phi_{\mu\nu} = \square \Phi_{\mu\nu} - (Q_{\mu\sigma}
-\nabla_\mu Q_\sigma) \Phi^\sigma{}_\nu - (Q_{\nu\sigma} -\nabla_\nu
Q_\sigma)\Phi^\sigma{}_\mu. \label{Identity:ReducedWave}
\end{eqnarray}
\end{subequations}

\subsubsection{The subsidiary gauge evolution system }

In the calculations of this section we make the following assumption:

\begin{assumption}
\label{Assumption:PropagationConstraints}
{\em Let $g_{\mu\nu}$ and $\Phi_{\mu\nu}$, with $\Phi_\mu{}^\mu=
g^{\mu\nu}\Phi_{\mu\nu}=0$, be smooth solutions to the equations
\begin{subequations}
\begin{eqnarray}
&& \mathscr{R}_{\mu\nu} = 2\Phi_{\mu\nu} +\frac{1}{4}\mathcal{R}(x)g_{\mu\nu}, \label{AssumptionEquation1}\\
&& \blacksquare \Phi_{\mu\nu} = 4 \Phi_\mu{}^\lambda \Phi_{\nu\lambda} -
\Phi_{\lambda\rho}\Phi^{\lambda\rho} g_{\mu\nu} + \frac{1}{3}\mathcal{R}(x)
\Phi_{\mu\nu} + \frac{1}{6}\nabla_\mu\nabla_\nu \mathcal{R}(x) -\square
\mathcal{R}(x)g_{\mu\nu}, \label{AssumptionEquation2}
\end{eqnarray}
\end{subequations}
for some smooth choice of the gauge source functions $\mathcal{F}^\mu(x)$ and
$\mathcal{R}(x)$. }
\end{assumption}

Combining equation \eqref{AssumptionEquation1} with identity
\eqref{Identity:ReducedRicci} one finds the relation 
\begin{equation}
R_{\mu\nu} = 2\Phi_{\mu\nu} + \frac{1}{4}\mathcal{R}(x) g_{\mu\nu} +
\nabla_{(\mu}Q_{\nu)},
\label{RicciTensorToSolutionAndGaugeFields}
\end{equation}
in which the reduced Ricci operator has been eliminated. The latter
implies, in turn, that
\[
R = \mathcal{R}(x) + \nabla^\mu Q_\mu,
\]
so that, in fact 
\begin{equation}
Q = \nabla^\mu Q_\mu.
\label{IdentityZeroQuantities1}
\end{equation}
Also, it follows from its definition that
\[
Q= Q_\mu{}^\mu.
\]
Moreover, substituting equation
\eqref{RicciTensorToSolutionAndGaugeFields} into the definition of
$Q_{\mu\nu}$, equation \eqref{QTensor}, one obtains the
relation 
\begin{equation}
Q_{\mu\nu} =\nabla_{(\mu} Q_{\nu)}. \label{IdentityZeroQuantities2}
\end{equation}
Taking the divergence of this last identity, commuting covariant
derivatives and using  expression
\eqref{RicciTensorToSolutionAndGaugeFields} to eliminate the
components of the Ricci tensor which appears after commuting
derivatives one obtains
\begin{equation}
\nabla^\mu Q_{\mu\nu} = \frac{1}{8} Q_\nu \mathcal{R} + Q^{\mu} \Phi_{\mu\nu} +
\frac{1}{4} Q^{\mu} \nabla_{\mu}Q_{\nu} + \frac{1}{2}
\square Q_\nu + \frac{1}{2} \nabla_{\nu}Q + \frac{1}{4}
Q^{\mu} \nabla_{\nu}Q_{\mu}. 
\label{Divergence:Qmunu}
\end{equation}

\begin{remark}
{\em Equations \eqref{IdentityZeroQuantities1} and
  \eqref{IdentityZeroQuantities2} show that the zero--quantities $Q$,
  $Q^\mu$ and $Q_{\mu\nu}$ are not independent of each other. In what
  follows we will regard $Q^\mu$ as the fundamental zero
  quantity. Clearly, if $Q^\mu=0$ then necessarily $Q=0$ and $Q_{\mu\nu}=0$.}
\end{remark}

\medskip
The construction of a suitable system of subsidiary equations for the
fields $Q$, $Q_\mu$ and $Q_{\mu\nu}$ makes use of the properties of
the \emph{Bach tensor} $B_{ab}$ ---see Appendix \ref{Appendix:BachTensor}. From the
definition of the Bach tensor given in equation
\eqref{BachTensor}  one can find an expression for $B_{ab}$
which is homogeneous in the fields $Q$, $Q_\mu$ and $Q_{\mu\nu}$:
\begin{eqnarray*}
&& B_{\mu\nu} = - \frac{5}{12} Q \Phi_{\mu\nu} -  \Phi_{\nu}{}^{\lambda}
   Q_{\mu\lambda} -
   \Phi_{\mu}{}^{\lambda} Q_{\nu\lambda} + \frac{1}{24} Q^2 g_{\mu\nu} -  \frac{1}{48} Q \mathcal{R}(x)
   g_{\mu\nu} -  \frac{5}{48} Q \nabla_{\mu}Q_{\nu} + \frac{1}{48} \mathcal{R}(x) \nabla_{\mu}Q_{\nu} \\
&& \hspace{2cm}+ 2 \Phi_{\nu\lambda} \nabla_{\mu}Q^{\lambda} -  \frac{1}{4} \nabla_{\mu}\nabla_{\lambda}\nabla^{\lambda}Q_{\nu} -  \frac{1}{16} Q \nabla_{\nu}Q_{\mu} 
+ \frac{1}{16} \mathcal{R}(x) \nabla_{\nu}Q_{\mu} + \frac{3}{16} \nabla_{\mu}Q^{\lambda}
   \nabla_{\nu}Q_{\lambda} \\
&& \hspace{2cm}+ \frac{7}{4} \Phi_{\mu\lambda} \nabla_{\nu}Q^{\lambda} +
   \frac{1}{6} \nabla_{\nu}\nabla_{\mu}Q + \frac{1}{4}
   \nabla_{\lambda}\nabla_{\mu}\nabla^{\lambda}Q_{\nu} + \frac{1}{12} g_{\mu\nu}
   \nabla_{\lambda}\nabla^{\lambda}Q -  \frac{1}{4}
   \nabla_{\lambda}\nabla^{\lambda}\nabla_{\mu}Q_{\nu} \\
&& \hspace{2cm}-  \frac{1}{4}
   \nabla_{\lambda}\nabla^{\lambda}\nabla_{\nu}Q_{\mu} + \frac{3}{4} \Phi_{\nu\lambda}
   \nabla^{\lambda}Q_{\mu} + \frac{1}{8} \nabla_{\nu}Q_{\lambda} \nabla^{\lambda}Q_{\mu} +
   \frac{1}{16} \nabla_{\lambda}Q_{\nu} \nabla^{\lambda}Q_{\mu} + \frac{1}{2}
   \Phi_{\mu \lambda} \nabla^{\lambda}Q_{\nu} \\
&& \hspace{2cm}+ \frac{1}{8} \nabla_{\mu}Q_{\lambda}
   \nabla^{\lambda}Q_{\nu} -  \frac{1}{2} \Xi d_{\mu\lambda\nu\rho} \nabla^{\rho}Q^{\lambda} -
   \frac{1}{4} \Xi d_{\mu\rho\nu\lambda} \nabla^{\rho}Q^{\lambda} -  \frac{3}{4}
   \Phi_{\lambda\rho} g_{\mu\nu} \nabla^{\rho}Q^{\lambda}\\
&& \hspace{2cm} -  \frac{1}{16} g_{\mu\nu}
   \nabla_{\lambda}Q_{\rho} \nabla^{\rho}Q^{\lambda} -  \frac{1}{16} g_{\mu\nu}
   \nabla_{\rho}Q_{\lambda} \nabla^{\rho}Q^{\lambda}.
\end{eqnarray*}

From the previous identity one finds, after some manipulations, that 
\[
\nabla^\mu B_{\mu\nu} =-\frac{1}{4}\square^2 Q_\nu+ H_{\nu}(\nabla \square
{\bm Q}, \nabla Q,\nabla {\bm Q}, {\bm Q}, Q).
\]
In view of the above, it is convenient to define the auxiliary
tensor field
\[
M_a\equiv \square Q_a.
\]
A further calculation then shows that
\[
\square Q = H(\nabla {\bm M}, \nabla Q,\nabla {\bm Q},Q).
\]

\medskip
Recalling that the Bach tensor is divergence free, ie. $\nabla^a
B_{ab}=0$, it follows from the discussion in the previous paragraph
that the fields $M_a$, $Q_a$ and $Q$ satisfy an homogeneous system of
wave equations of the form
\begin{subequations}
\begin{eqnarray}
&& \square M_\mu =4 H_{\mu}(\nabla
{\bm M}, \nabla Q,\nabla {\bm Q}, {\bm Q}, Q), \label{GaugeSubsidiarySystem1}\\
&& \square Q_\mu =M_\mu, \label{GaugeSubsidiarySystem2}\\
&& \square Q = H(\nabla {\bm M}, \nabla Q,\nabla {\bm Q},Q). \label{GaugeSubsidiarySystem3} 
\end{eqnarray}
\end{subequations}
In the following, the above system will be known as \emph{gauge
  subsidiary evolution system}. Given the homogeneous nature of
\eqref{GaugeSubsidiarySystem1}-\eqref{GaugeSubsidiarySystem3}, if the
system is supplemented with vanishing boundary and initial conditions,
one necessarily has the unique solution
\[
M_\mu=0, \qquad Q_\mu=0, \qquad Q=0, \qquad \mbox{in a neighbourhood of} 
\ \partial\mathcal{S}_\star. 
\]
The latter, in turn, implies that
\[
Q_{\mu\nu}=0 \qquad \mbox{in a neighbourhood of} \ \partial\mathcal{S}_\star. 
\]

\begin{remark}
{\em If this is the case, then, at least in a neighbourhood of
  $\partial \mathcal{S}_\star$ one has that 
\[
R = \mathcal{R}(x), \qquad \Gamma^\mu =\mathcal{F}^\mu(x)
\]
and the tensor $\Phi_{ab}$ coincides with one half of the tracefree
part of the Ricci tensor, $R_{ab}$, of the metric $g_{ab}$.  
 }
\end{remark}

\subsubsection{Initial and Boundary conditions for the subsidiary
  gauge evolution system}
In this section we analyse the trivial initial conditions
\[
M_{\mu}=0, \qquad Q_{\mu}=0, \qquad Q = 0, \qquad \nabla_\mu M_\nu = 0,\qquad 
\nabla_\mu Q_\nu=0, \qquad \nabla_\mu Q = 0 \qquad \mbox{on}\quad \mathcal{S}_\star
\]
and the trivial boundary conditions 
\[
M_{\mu}=0, \qquad Q_{\mu}=0, \qquad Q=0 \qquad \mbox{on} \quad \mathscr{I}
\]
and consider the conditions under which they can be enforced.

\medskip
In order to study the consequences of these vanishing initial-boundary
conditions, it is convenient to decompose $Q_\mu$ in terms of its intrinsic and
normal components. In the case of $\mathscr{I}$, the projections $\hat{q}_\mu
\equiv \ell_\mu{}^\nu Q_\nu$ and $\hat{q} \equiv \notn^\nu Q_\nu$ are naturally
introduced. The fundamental zero-quantity is then written as
\[
Q_\mu \simeq \hat{q}_\mu +
\hat{q} \notn_\mu.
\]
Adopting a Gaussian gauge as in Section \ref{Section:GeneralSetUp}, the conditions $\hat{q}_\mu
\simeq 0$ and $\hat{q} \simeq 0$ imply a system of equations for the normal
derivatives of the lapse and shift ---see equations
\eqref{GeneralisedWaveCoordinatesDecomposed1} and
\eqref{GeneralisedWaveCoordinatesDecomposed2}. Namely, one has that 
\begin{equation}
\partial_1 \notalpha \simeq 3\varkappa + \mathcal{F}^1, \qquad 
\partial_1\notbeta^\delta \simeq \emph{F}^\delta - \gamma^\delta.
\label{QGauge1}
\end{equation}
Additionally, when the conditions $Q \simeq 0$ and $M_\mu \simeq 0$ are
imposed, the following relations are found:
\begin{equation}
\notD \hat{q} \simeq 0, \qquad \notD^2\hat{q} \simeq 0, \qquad
\notD^2\hat{q}_\mu + \varkappa \notD\hat{q}_\mu \simeq 0.
\label{QGauge2}
\end{equation}
These can be read as higher order differential equations for the normal
derivatives of $\notalpha$ and $\notbeta$.

\medskip

Following the same approach, in the case of $\mathcal{S}_\star$ one defines the
projections $q_\mu \equiv h_\mu{}^\nu Q_\nu$ and $q \equiv n^\nu Q_\nu$, so one
has
\[
Q_\mu = q_\mu - q n_\mu.
\]
Setting $Q_\mu =0$, an analogous decomposition of the metric implies a
pair of evolution equations for $\alpha$ and $\beta$ in terms of the gauge
source functions. When the remaining vanishing initial data are analysed, a
series of straightforward calculations leads to the following conditions on $q$
and $q_\mu$:
\begin{equation}
q = 0, \quad q_\mu = 0, \quad D^{(n)}q = 0, \quad D^{(n)}q_\mu = 0, \quad n = 1, 2, 3.
\label{QGauge3}
\end{equation}
Therefore, more restrictions in the form of higher order constraints for the
lapse and shift functions are imposed.

\begin{remark} \em{Even though the condition $Q = 0$ on
$\mathcal{S}_\star$ and $\mathscr{I}$ implies, respectively, that $Dq = 0$ and
$\notD\hat{q} \simeq 0$, this can be equivalently stated as imposing that the
function $\mathcal{R}(x)$ coincides with the Ricci scalar of the metric $g_{ab}$.}
\end{remark}

The discussion of the section can be summarised in the following lemma:

\begin{lemma}
Let $Q, \ Q_\mu$ and $Q_{\mu\nu}$ be defined as in \eqref{QScalar}--\eqref{QTensor}.
If conditions \eqref{QGauge1} and \eqref{QGauge2} are satisfied on $\mathscr{I}$,
and \eqref{QGauge3} is satisfied on $\mathcal{S}_\star$, then $Q, \ Q_{\mu}$ and
$Q_{\mu\nu}$ vanish identically in a neighbourhood of $\partial\mathcal{S}_\star$.
\end{lemma}

\section{The local existence result}
\label{Section:LocalExistenceResult}

We are now in the position of formulating the main result of this article: a
local in time existence result for the conformal Einstein field equations in a
neighbourhood of the corner $\partial\mathcal{S}_\star$. This result can, in
turn, be patched together with the domain of dependence of open subsets of
$\mathcal{S}_\star$ away from $\partial\mathcal{S}_\star$ to obtain a solution
on a slab around $\mathcal{S}_\star$ ---see e.g. \cite{CFEBook}, Section 12.3. 

\medskip
One has the following:

\begin{theorem}
\label{Theorem:Main}
Let $\mathcal{S}_\star$ be a 3-dimensional spacelike hypersurface with boundary
$\partial\mathcal{S}_\star$ and smooth anti-de Sitter-like initial data defined on it.
Consider the cylinder $[0, \tau_\bullet) \times \partial\mathcal{S}_\star$, for some
$\tau_\bullet > 0$, endowed with a smooth 3-dimensional Lorentzian metric $\ell_{ij}$
and let $\psi_0, \ \psi_4$ be two complex-valued scalar functions. Assume that
the data on $\mathcal{S}_\star$ and the cylinder satisfy the corner conditions
at $\partial\mathcal{S}_\star$. Then, there exists a smooth solution to the Einstein field
equations with $\lambda < 0$ in a neighbourhood of $\mathcal{S}_\star$.

\end{theorem}

\begin{proof}

Consider initial data on $\mathcal{S}_\star$ given as in Definition
\ref{Definition:AdSData}. Given a 3--dimensional Lorentzian metric $\ell_{ij}$
on the cylinder $[0, \tau_\bullet) \times \partial\mathcal{S}_\star$, the data
given by \eqref{BoundaryPrescription:SigmaFriedrich},
\eqref{BoundaryPrescription:Schouten} and
\eqref{BoundaryPrescription:MagneticWeyl} can be computed. On the other hand,
$\notd_{ij}$ is determined via the system
\eqref{GaussConstraintDecomposed1}--\eqref{GaussConstraintDecomposed2}, which
requires the specification of $\psi_0$ and $\psi_4$ along with initial values
for $w$ and $w_i$. Notice that the latter ones are prescribed by the initial
data at $\partial\mathcal{S}_\star$. If these two sets of initial and boundary
data satisfy the corner conditions at $\partial\mathcal{S}_\star$ then the
theory of initial-boundary value problems, as given in e.g.
\cite{CheWah83,DafHru85}, guarantees the existence of a unique solution to the
system of wave equations \eqref{ReducedWaveCFE1}--\eqref{ReducedWaveCFE5} in a
neighbourhood of $\partial\mathcal{S}_\star$.

\medskip
Given the boundary data described above, Proposition
\ref{Proposition:SummaryBoundaryData} and Lemma \ref{SubsidiaryEquations} imply
that all the components of the zero-quantities vanish on the cylinder. On the
other hand, from \Cref {VanishingZQInitialData} and
\Cref{VanishingDerivativesZQInitialData} we have that the initial data on
$\mathcal{S}_\star$ yield vanishing data for the zero--quantities and their
first-order derivatives on this hypersurface. Thus, Proposition
\ref{PropagationZQ} implies that a solution to the system
\eqref{ReducedWaveCFE1}--\eqref{ReducedWaveCFE5} guarantees the existence and
uniqueness of a vanishing solution of equations
\eqref{SubsidiaryEquation1}--\eqref{SubsidiaryEquation4}. From the definition
of the zero--quantities it follows then that the conformal Einstein field
equations \eqref{CFE1}--\eqref{CFE5} are satisfied in a neighbourhood of
$\partial\mathcal{S}_\star$.

\medskip
Finally, having a solution to the conformal Einstein field equations,
Proposition \ref{FriedrichThm} implies that the metric $\tilde{g}_{ab} =
\Xi^{-2}g_{ab}$ is a solution to the Einstein field equations \eqref{EFE} with
$\lambda < 0$ for $\Xi \neq 0$. 

\end{proof}

\begin{remark}
{\em A more precise statement about the regularity of the initial data and
boundary conditions needs to be expressed in terms of suitable Sobolev spaces
and goes beyond the scope of this article. Here, for the sake of simplicity of
the presentation we have opted for phrase these conditions in terms of the word
\emph{smooth}.}
\end{remark}

\section{Conclusions}
\label{Section:Conclusions}

The construction carried out in this work can be implemented to numerical codes
in a systematic way. Moreover, it represents a step forward with respect to the
work in \cite{Fri95} as the wave equations to be solved are manifestly
hyperbolic.  Furthermore, the free boundary data for the Weyl tensor are
explicitly related to the incoming and outgoing radiation.  This may make
possible to study more general boundary conditions, relevant for a better
understanding of the instability of anti-de Sitter-like spacetimes.
Nevertheless, as the geometric character of these data is broken by the
performed decompositions, further work must be done in order to obtain a
completely covariant formulation.

The local existence result presented assumes a vanishing matter--energy
component. However, under the methods of conformal geometry, it is not clear
how completely general scenarios can be studied ---see \cite{Fri14c} for a
discussion of a particular case of the Einstein-massive scalar field case which
is particularly amenable to the use of conformal methods and \cite{Fri17} for a
discussion about dust models coupled to Einstein equations with $\lambda >0$.
Despite this fact, if a \emph{tracefree} matter field is considered, it is
possible to establish a well--posed problem and analyse it in a similar fashion
to the one described here. In this context, work is currently under progress to
investigate a possible extension of the result to anti-de Sitter-like
spacetimes for this class of energy--momentum tensors.  Moreover, this
naturally leads to several particular cases of interest, namely, Maxwell,
Yang--Mills and scalar fields.

\section*{Acknowledgements}
JAVK thanks the hospitality of Erwin Schr\"odinger Institute for
Mathematics and Physics of the University of Vienna during the
programme \emph{Geometry and Relativity} which took place between the
months of July and September of 2017. DAC thanks support granted by
CONACyT (480147). The calculations described in this article have been
carried out in the suite {\tt xAct} for {\tt Mathematica} ---see
\cite{xAct}.

\appendix

\section{Zero-quantities and integrability conditions}
\label{Appendix:ZeroQuantities}

The zero-quantities defined in Section \ref{zero_quantities} possess a set
of \emph{integrability conditions} which allow us to understand the
interrelations between the various equations. These conditions
arise in a natural way when appropriate antisymmetrisations of the derivatives
of the zero-quantities are taken into consideration. 

\subsection{Basic properties}
First, the zero--quantities possess the following symmetries:
\begin{align}
\begin{split}
\Upsilon_{ab} = \Upsilon_{(ab)}, \quad \Delta_{abc} = \Delta_{[ab]c}, \quad
\Delta_{[abc]} = 0, \quad \Lambda_{abc} = \Lambda_{a[bc]}, \quad \Lambda_{[abc]} =0.
\end{split}
\end{align}
Regarding $\Lambda_{abc}$, it will prove to be useful to define the following
auxiliary zero--quantity:
\begin{eqnarray}
&& \Lambda_{abcde} \equiv 3\nabla_{[a}d_{bc]de} = \Lambda_{d[ab} g_{c]e}  -\Lambda_{e[ab} g_{c]d}.
\label{Lambda_abcde}
\end{eqnarray}
Notice that the vanishing of this new quantity expresses, in an alternative
way, the validity of equation \eqref{CFE4} ---see \cite{CarHurVal18}.  Also, observe that
$\Lambda_{abcde}$ has only two independent divergences. From its definition,
and assuming the validity of the wave equation \eqref{WaveCFE4}, it follows
that
\begin{subequations}
\begin{eqnarray}
&& \nabla^c\Lambda_{abcde} = 0, \label{Lambda_abcdeDivergence1} \\
&& \nabla^e\Lambda_{abcde} = 3\nabla_{[a}\Lambda_{|d|bc]} \label{Lambda_abcdeDivergence2}.
\end{eqnarray}
\end{subequations}

Using equations \eqref{Lambda_abcdeDivergence1},
\eqref{Lambda_abcdeDivergence2} and the wave equations
\eqref{WaveCFE1}--\eqref{WaveCFE4}, one can check via some suitable
contractions that the zero--quantities exhibit some useful properties, namely:

\begin{equation}
\begin{gathered}
 \Upsilon_a{}^a = 0, \quad \nabla_{b}\Upsilon_a{}^b = 3\Theta_a, \quad \nabla_a\Theta^a = \Upsilon^{ab}L_{ab}, 
\quad \Delta_a{}^b{}_b  = 0, \quad \Lambda^b{}_{ab} = 0, \\
 \quad \nabla_b\Delta_a{}^b{}_c = \Lambda_{acb}\nabla^b\Xi, \quad \nabla_c\Delta_{ab}{}^c = 
 - \Lambda_{cab}\nabla^c\Xi, \quad
\nabla_a \Lambda^a{}_{bc} = 0, \quad \nabla_b\Lambda_{a}{}^b{}_{c} = 0, \quad 
\Delta_{abc}= \Xi\Lambda_{cba}.
\label{ZQIdentities}
\end{gathered}
\end{equation}

Finally, it is noticed that the expressions for $\Upsilon_a{}^a, \
\nabla_a\Theta^a, \ \nabla_b\Delta_a{}^b{}_c$ and $\nabla^c\Lambda_{abcde}$
given above represent the geometric wave equations
\eqref{WaveCFE1}-\eqref{WaveCFE4}, respectively.

\subsection{Integrability conditions}
Making use of these identities, direct computations yield the following integrability conditions for the
zero-quantities associated to vacuum conformal Einstein field equations:
\begin{subequations}
\begin{eqnarray}
&& 2\nabla_{[a}\Upsilon_{b]c}  = 2g_{c[a} \Theta_{b]} + \Xi \Delta_{abc}, \label{IC1} \\
&& 2\nabla_{[a}\Theta_{b]}  = -  2L_{[a}{}^{c} \Upsilon_{b]c} +
   \Delta_{abc} \nabla^{c}\Xi, \label{IC2}\\
&& \nabla_{[a}\Delta_{bc]d} = 
d_{de[ab} \Upsilon^{e}{}_{c]} + \Lambda_{d[ab} \nabla_{c]}\Xi - \Lambda_{e[ab}g_{c]d}\nabla^e\Xi, \label{IC3}\\
&& \nabla_{[a} \Lambda_{b]cd} = 0. \label{IC4}
\end{eqnarray}
\end{subequations}
Here, the last integrability condition is a direct consequence of equation
\eqref{Lambda_abcdeDivergence1}.

\begin{remark}

{\em 
One can check that equation \eqref{Lambda_abcdeDivergence2} also
represents an integrability condition for $\Lambda_{abc}$.
}

\end{remark}

\subsection{The subsidiary equations}
With the previous equations at hand, and assuming that the fields $\Xi, s,
L_{ab}$ and $d_{abcd}$ satisfy the wave equations
\eqref{WaveCFE1}--\eqref{WaveCFE4}, suitable evolution (i.e. wave) equations
for the zero-quantities can be obtained by direct application of a covariant
derivative.  Commuting derivatives, and using the identities exposed in A.1,
straightforward but lengthy calculations yield:
\begin{eqnarray*}
&& \nabla_{d}\nabla^{d}Z_{ab} = - \Theta Z^{dc} d_{adbc} + \tfrac{1}{6} \
Z_{ab} R + Z_{b}{}^{d} L_{ad} + 3 Z_{a}{}^{d} L_{bd} - 2 Z^{dc} \
L_{dc} g_{ab} + \nabla_{a}Z_{b} + 3 \nabla_{b}Z_{a} -  \Theta \
\Lambda_{dab} \nabla^{d}\Theta, \\
&& \nabla_{b}\nabla^{b}Z_{a} = 6 L_{ab} Z^{b} -  Z^{bc} \Delta_{abc} + 2 \Theta L^{bc} \Delta_{abc} 
+ Z^{bc} \nabla_{a}L_{bc} -  \tfrac{1}{6} Z_{ab} \nabla^{b}R + Z^{bc} \nabla_{c}L_{ab}, \\
&& \nabla_{d}\nabla^{d}\Delta_{abc} = 2 Z_{[a}{}^{d} \Lambda_{|c|b]d} 
-  Z_{c}{}^{d} \Lambda_{dab} + 2s \Lambda_{cab} + 2 \Theta^2 \Lambda_{c}{}^{de} d_{adbe} 
+ 2\Theta^2 \Lambda^{d}{}_{[a}{}^{e} d_{|cd|b]e} 
- \tfrac{1}{2} \Theta \Lambda_{cab} R  \\
&& \hspace{2cm} + 3 d_{cdab} Z^{d} + 2Z^{de} \Lambda_{de[b} g_{a]c}
+ Z^{de} \nabla_{e}d_{cdab} + 2d_{cde[a} \nabla^{e}Z_{b]}{}^{d}, \\
&& \nabla_{d}\nabla^{d}\Lambda_{abc} = - 2\Theta \Lambda^{d}{}_{[b}{}^{e} d_{c]ead} 
+ \tfrac{1}{6} \Lambda_{mbc} R - 2\Lambda^{d}{}_{a[b} L_{c]d} + 3 \Lambda^{d}{}_{bc} L_{ad} 
+ 2g_{a[b}\Lambda^{d}{}_{c]}{}^{e} L_{de}.
\end{eqnarray*}
These equations are clearly homogeneous in the zero-quantities and their first
derivatives.

\subsection{The Bach tensor}
\label{Appendix:BachTensor}

The Bach tensor in 4-dimensions is defined as
\begin{equation}
\label{BachTensor}
B_{ab} \equiv \nabla^c\nabla_a L_{bc} - \nabla^c\nabla_c L_{ab} - C_{acbd}L^{cd}.
\end{equation}
It can be written in terms of zero-quantities as
\begin{equation}
\label{BachTensorZeroQuantities}
B_{ab} = \nabla^c \Delta_{acb} -\Lambda_{abc}\nabla^c\Xi - d_{acbd} Z^{cd}.
\end{equation}
Thus, for any solution to the conformal Einstein field equations one has
that $B_{ab}=0$. Finally, it is observed that
\[
\nabla^a B_{ab}=0
\]
independently of whether the conformal Einstein field equations hold
or not.

\end{document}